\documentclass[11pt]{article}

\usepackage[margin=1in]{geometry}

\usepackage{hyperref}       
\usepackage{url}            
\usepackage{amsfonts}       
\usepackage{nicefrac}       
\usepackage{microtype}      

\usepackage{algorithm}
\usepackage[noend]{algorithmic}

\usepackage{ifthen}
\usepackage{xspace}
\usepackage{xstring}
\usepackage{amssymb,amsmath}
\usepackage[mathscr]{euscript}

\newcommand{\Eat}[1]{}

\newcommand{\RNum}[1]{\uppercase\expandafter{\romannumeral #1\relax}}

\newcommand{\SetCard}[1]{\ensuremath{| #1 |}}

\DeclareSymbolFont{AMSb}{U}{msb}{m}{n}
\DeclareMathSymbol{\N}{\mathord}{AMSb}{"4E}
\DeclareMathSymbol{\Z}{\mathord}{AMSb}{"5A}
\DeclareMathSymbol{\R}{\mathord}{AMSb}{"52}

\newcommand{\poly}{\ensuremath{\mbox{poly}}\xspace}

\DeclareMathOperator*{\conv}{conv}

\newcommand{\InsertAlgorithm}[3]%
{\begin{algorithm}[ht]
\caption{\sc #1}\label{#2}
\begin{algorithmic}[1]
\vspace{0.1cm}
\baselineskip=1.1\baselineskip
#3
\end{algorithmic}\end{algorithm}}

\newcommand{\AlgAssign}{\ensuremath{ \leftarrow }\xspace}
 
\newcommand{\Ordinal}[1]{\ensuremath{{#1}^{\rm th}}}
\newcommand{\Set}[1]{\ensuremath{\{ #1 \}}}
\newcommand{\SpecSet}[2]{\ensuremath{\Set{#1 \mid #2}}}

\newcommand{\Floor}[1]{\ensuremath{\left\lfloor{#1}\right\rfloor}}
\newcommand{\Angle}[1]{\ensuremath{\langle{#1}\rangle}}
\newcommand{\Size}[1]{\ensuremath{| #1 |}}

\newcommand{\Compl}[1]{\overline{#1}}

\def\half{\ensuremath{\frac{1}{2}}\xspace}

\providecommand{\Kth}[1]{\ensuremath{{#1}^{\rm th}}}

\newcommand{\Complexity}[1]{\ensuremath{\text{#1}}}
\newcommand{\Hardness}[2]{\ensuremath{\text{#1-#2}}}

\newtheorem{definition}{Definition}[section]
\newtheorem{problem}{Problem}[section]

\newtheorem{theorem}{Theorem}[section] 
\newtheorem{lemma}[theorem]{Lemma}

\newtheorem{corollary}[theorem]{Corollary}
\newtheorem{proposition}[theorem]{Proposition}
\newtheorem{observation}[theorem]{Observation}

\def\QED{{\phantom{x}} \hfill \ensuremath{\rule{1.3ex}{1.3ex}}}

\newcommand{\papertitle}{A General Framework for Robust Interactive Learning}

\title{\papertitle}

\author{%
Ehsan Emamjomeh-Zadeh%
\thanks{%
Department of Computer Science,
University of Southern California,
emamjome@usc.edu} 
\and
David Kempe%
\thanks{%
Department of Computer Science,
University of Southern California,
dkempe@usc.edu}}

\begin{document}

\maketitle

\def\structure{model\xspace}
\def\Structure{Model\xspace}
\def\structures{models\xspace}
\def\Structures{Models\xspace}

\def\Bsort{\textsc{Bubble Sort}\xspace}
\def\Isort{\textsc{Insersion Sort}\xspace}
\def\AnySwapText{Arbitrary Swap\xspace}
\def\AnySwap{\textsc{\AnySwapText}\xspace}

\def\resformat{format\xspace}
\def\Resformat{Format\xspace}
\def\resformats{formats\xspace}
\def\Resformats{Formats\xspace}

\def\length{length\xspace}
\def\Length{Length\xspace}
\def\lengths{lengths\xspace}
\def\Lengths{Lengths\xspace}

\def\weightV{weight\xspace}
\def\WeightV{Weight\xspace}
\def\weightsV{weights\xspace}
\def\WeightsV{Weights\xspace}


\newcommand{\tvd}[2]{\ensuremath{d_{\mathrm{TV}}(#1, #2)}\xspace}
\newcommand{\Entr}[1]{\ensuremath{H(#1)}\xspace}

\newcommand{\Weight}[1][]{\ensuremath{%
\ifthenelse{\equal{#1}{}}{\omega}{\omega_{#1}}}\xspace}
\newcommand{\Reach}[3][]{\ensuremath{%
\ifthenelse{\equal{#1}{}}{N(#2,#3)}{N_{#1}(#2,#3)}}\xspace}
\def\PotentialSym{\ensuremath{\Phi}\xspace}
\newcommand{\Potential}[2]{%
\ensuremath{\PotentialSym_{#1}(#2)}\xspace}
\newcommand{\MedianPotential}[2]{%
\ensuremath{\Psi_{#1}(#2)}\xspace}
\def\TVDis{\ensuremath{\Delta}\xspace}

\newcommand{\AllStructures}{\ensuremath{\Sigma}\xspace}
\def\ConsistentSet{\ensuremath{S}\xspace}
\def\InitSet{\ensuremath{\ConsistentSet_{\text{init}}}\xspace}

\def\TOTQ{\ensuremath{K}\xspace}
\def\ErrProb{\ensuremath{\delta}\xspace}
\def\SampleErr{\ensuremath{\varepsilon}\xspace}
\newcommand\NodeWeight[1][]{\ensuremath{%
\ifthenelse{\equal{#1}{}}{\mu}{\mu(#1)}}\xspace}
\newcommand\NodeWeightP[1][]{\ensuremath{%
\ifthenelse{\equal{#1}{}}{\mu'}{\mu'(#1)}}\xspace}
\def\Tolerance{\ensuremath{\lambda}\xspace}
\def\Marked{\ensuremath{M}\xspace}
\def\MultiWeights{\mbox{\ensuremath{%
\text{\sc Multiplicative-Weights}}}\xspace}
\newcommand{\WPotential}[2][]{\ensuremath{%
\ifthenelse{\equal{#1}{}}{\Phi(#2)}{\Phi_{#1}(#2)}}\xspace}
\newcommand{\SumWeight}[2][]{\ensuremath{%
\ifthenelse{\equal{#1}{}}{\Gamma(#2)}{\Gamma_{#1}(#2)}}\xspace}
\def\AlgThreshold{\ensuremath{\tau}\xspace}

\newcommand{\EventProb}[2]{\ensuremath{P_{#1, #2}}\xspace}
\def\NumSample{\ensuremath{r}\xspace}


\newcommand{\CLG}{\ensuremath{G_{\text{CL}}}\xspace}

\def\Family{\ensuremath{\mathcal{F}}\xspace}
\def\Classifier{\ensuremath{C}\xspace}
\def\ClassifierP{\ensuremath{C'}\xspace}
\def\Dim{\ensuremath{d}\xspace}
\def\Hyperplanes{\ensuremath{\mathcal{H}}\xspace}
\newcommand{\HyperplanesS}[2]{\ensuremath{\mathcal{H}^{#1}_{#2}}\xspace}


\newcommand{\BSG}{\ensuremath{G_{\text{BS}}}\xspace}
\newcommand{\ISG}{\ensuremath{G_{\text{IS}}}\xspace}
\newcommand{\PER}{\ensuremath{\pi}\xspace}
\newcommand{\PERP}{\ensuremath{\pi'}\xspace}
\newcommand{\Per}[1]{\ensuremath{\PER_{#1}}\xspace}
\newcommand{\PerP}[1]{\ensuremath{\PERP_{#1}}\xspace}
\newcommand{\sorted}{\ensuremath{\PER^*}\xspace}

\def\AlgSym{\ensuremath{\mathcal{A}}\xspace}
\def\AdvSym{\ensuremath{\mathcal{B}}\xspace}
\def\ListSym{\ensuremath{L}\xspace}

\def\QC{\ensuremath{Q}\xspace}

\def\PosLabel{\ensuremath{\text{`\ensuremath{+}'}}\xspace}
\def\NegLabel{\ensuremath{\text{`\ensuremath{-}'}}\xspace}

\newcommand{\ShiftPer}[3]{\ensuremath{{#1}_{#3 \leftarrow #2}}\xspace}

\def\LinearExt{\ensuremath{\mathcal{L}}\xspace}


\newcommand{\GCG}{\ensuremath{G_{\text{GC}}}\xspace}
\newcommand{\UCG}{\ensuremath{G_{\text{UC}}}\xspace}

\def\Split{\ensuremath{\mbox{\textsc{Split}}}\xspace}
\def\Merge{\ensuremath{\mbox{\textsc{Merge}}}\xspace}
\def\Clustering{\ensuremath{\mathcal{C}}\xspace}
\def\ClusteringOpt{\ensuremath{\mathcal{C}^*}\xspace}
\def\ClusteringP{\ensuremath{\mathcal{C}'}\xspace}
\def\ClusteringPP{\ensuremath{\mathcal{C}''}\xspace}
\def\ClusteringBar{\ensuremath{\bar{\mathcal{C}}}\xspace}
\def\OptClustering{\ensuremath{\mathcal{C^{*}}}\xspace}

\newcommand{\Adjacency}[1][]{\ensuremath{%
\ifthenelse{\equal{#1}{}}{\mathcal{A}}{\mathcal{A}^{(#1)}}}\xspace}
\newcommand{\Diff}[2]{\ensuremath{d(#1, #2)}\xspace}

\begin{abstract}
We propose a general framework for interactively learning \structures,
such as (binary or non-binary) classifiers, orderings/rankings of items,
or clusterings of data points.
Our framework is based on a generalization of
Angluin's equivalence query model
and Littlestone's online learning model:
in each iteration, the algorithm proposes a \structure,
and the user either accepts it
or reveals a specific mistake in the proposal.
The feedback is correct only with probability $p > \half$
(and adversarially incorrect with probability $1 - p$),
i.e., the algorithm must be able to learn
in the presence of arbitrary noise.
The algorithm's goal is to learn the ground truth \structure
using few iterations.

Our general framework is based on a graph representation
of the \structures and user feedback.
To be able to learn efficiently, it is sufficient that there be
a graph $G$ whose nodes are the \structures and (weighted) edges
capture the user feedback, with the property that if $s, s^*$
are the proposed and target \structures, respectively,
then any (correct) user feedback $s'$ must lie
on a shortest $s$-$s^*$ path in $G$.
Under this one assumption, there is a natural algorithm
reminiscent of the Multiplicative Weights Update algorithm,
which will efficiently learn $s^*$
even in the presence of noise in the user's feedback.

From this general result, we rederive with barely any extra effort
classic results on learning of classifiers
and a recent result on interactive clustering;
in addition, we easily obtain new interactive learning algorithms
for ordering/ranking.

\end{abstract}

\section{Introduction} \label{sec:introduction}

With the pervasive reliance on machine learning systems
across myriad application domains in the real world,
these systems frequently need to be deployed
before they are fully trained.
This is particularly true when the systems are supposed
to learn a specific user's (or a small group of users')
personal and idiosyncratic preferences.
As a result, we are seeing an increased practical interest
in online and interactive learning across a variety of domains.

A second feature of the deployment of such systems ``in the wild''
is that the feedback the system receives is likely to be noisy.
Not only may individual users give incorrect feedback,
but even if they do not, the preferences --- and hence feedback ---
across different users may vary.
Thus, interactive learning algorithms deployed
in real-world systems must be resilient to noisy feedback.

Since the seminal work of
Angluin~\cite{angluin:1988:queries-concept}
and Littlestone~\cite{littlestone:1988:online-learning},
the paradigmatic application of (noisy) interactive learning has been
online learning of a binary classifier when the algorithm is provided
with feedback on samples it had previously classified incorrectly.
However, beyond (binary or other) classifiers,
there are many other \structures
that must be frequently learned in an interactive manner.
Two particularly relevant examples are the following:

\begin{itemize}
\item Learning an ordering/ranking of items is a key part
of personalized Web search or other information-retrieval systems
(e.g., \cite{joachims:2002:search-engine-clickthrough,%
radlinski-joachims:2005:query-chain}).
The user is typically presented with an ordering of items,
and from her clicks or lack thereof,
an algorithm can infer items that are in the wrong order.

\item
Interactively learning a clustering
\cite{balcan-blum:2008:split-merge,%
awasthi-zadeh:2010:supervised-clustering,%
awasthi-balcan-voevodski:2017:local-algorithm-journal}
is important in many application domains,
such as interactively identifying communities in social
networks or partitioning an image into distinct objects.
The user will be shown a candidate clustering,
and can express that two clusters should be merged,
or a cluster should be split into two.
\end{itemize}

In all three examples --- classification, ranking, and clustering ---
the interactive algorithm will propose
a \emph{\structure}\footnote{We avoid
the use of the term ``concept,'' as it typically refers
to a binary function, and is thus associated
specifically with a classifier.}
(a classifier, ranking, or clustering) as a solution.
The user then provides --- explicitly or implicitly --- feedback
on whether the \structure is correct or needs to be fixed/improved.
This feedback may be incorrect with some probability.
Based on the feedback, the algorithm will propose a new
and possibly very different \structure, and the process repeats.
This type of interaction is the natural generalization of
Angluin's equivalence query model%
~\cite{angluin:1988:queries-concept,angluin:1992:colt-survey}.
It is worth noting that in contrast to active learning,
in interactive learning (which is the focus of this work),
the algorithm cannot ``ask'' direct questions;
it can only propose a \structure and receive feedback in return.
The algorithm should minimize the number of user interactions,
i.e., the number of times that the user needs to propose fixes.
A secondary goal is to make the algorithm's internal computations
efficient as well.

The main contribution of this article is a general framework
for efficient interactive learning of \structures
(even with noisy feedback),
presented in detail in Section~\ref{sec:preliminaries}.
We consider the set of all $N$ \structures as nodes
of a positively weighted undirected or directed graph $G$.
The one key property that $G$ must satisfy is the following:
(*) If $s$ is a proposed \structure,
and the user (correctly) suggests changing it to $s'$,
then the graph must contain the edge $(s, s')$;
furthermore, $(s, s')$ must lie on a shortest path
from $s$ to the target \structure $s^*$ (which is unknown to the algorithm).

We show that this single property is enough to learn the target
\structure $s^*$ using at most $\log N$ queries\footnote{
Unless specified otherwise, all logarithms are base 2.}
to the user, in the absence of noise.
When the feedback is correct with probability $p > \half$,
the required number of queries gracefully deteriorates to $O(\log N)$;
the constant depends on $p$.
We emphasize that the assumption (*) is not an assumption
on the user. We do not assume that the user somehow ``knows''
the graph $G$ and computes shortest paths in order to find a response.
Rather, (*) states that $G$ was correctly chosen to model
the underlying domain, so that correct answers by the user
must in fact have the property (*).
To illustrate the generality of our framework,
we apply it to ordering, clustering, and classification:

\begin{enumerate}
\item For ordering/ranking, each permutation is a node in $G$;
one permutation is the unknown target.
If the user can point out only \emph{adjacent} elements
that are out of order, then $G$ is
an adjacent transposition ``\Bsort'' graph,
which naturally has the property (*).
If the user can pick any element and suggest that it should precede
an entire block of elements it currently follows,
then we can instead use an ``\Isort'' graph;
interestingly, to ensure the property (*),
this graph must be weighted.
On the other hand, as we show in Section \ref{sec:sorting},
if the user can propose two arbitrary elements that should be swapped,
there is \emph{no} graph $G$ with the property (*).

Our framework directly leads to an interactive algorithm
that will learn the correct ordering of $n$ items
in $O(\log (n!)) = O(n \log n)$ queries;
we show that this bound is optimal under
the equivalence query model.

\item For learning a clustering of $n$ items,
the user can either propose merging two clusters,
or splitting one cluster.
In the interactive clustering model of
\cite{balcan-blum:2008:split-merge,%
awasthi-zadeh:2010:supervised-clustering,%
awasthi-balcan-voevodski:2017:local-algorithm-journal},
the user can specify \emph{that}
a particular cluster $C$ should be split,
but does not give a specific split.
We show in Section~\ref{sec:clustering} that
there is a weighted directed graph with the property (*);
then, if each cluster is from a ``small'' concept class
of size at most $M$ (such as having low VC-dimension),
there is an algorithm finding the true clustering in
$O(k \log M)$ queries,
where $k$ is number of the clusters (known ahead of time).

\item For binary classification,
$G$ is simply an $n$-dimensional hypercube
(where $n$ is the number of sample points that are to be classified).
As shown in Section~\ref{sec:classification},
one immediately recovers a close variant of
standard online learning algorithms within this framework.
An extension to classification with more than two classes
is very straightforward.
\end{enumerate}


\section{Learning Framework} \label{sec:preliminaries}

We define a framework for query-efficient interactive learning
of different types of \emph{\structures}.
Some prototypical examples of \structures to be learned are
rankings/orderings of items,
(unlabeled) clusterings of graphs or data points,
and (binary or non-binary) classifiers.
We denote the set of all candidate \structures
(permutations, partitions,
or functions from the hypercube to $\Set{0, 1}$)
by \AllStructures,
and individual \structures\footnote{When considering specific
applications, we will switch to notation more in line with that used
for the specific application.} by $s, s', s^*$, etc.
We write $N = \Size{\AllStructures}$ for the number of
candidate \structures.

We study interactive learning of such \structures
in a natural generalization of
the equivalence query model
of Angluin~\cite{angluin:1988:queries-concept,%
angluin:1992:colt-survey}.
This model is equivalent
to the more widely known online learning model
of Littlestone~\cite{littlestone:1988:online-learning},
but more naturally fits the description of
user interactions we follow here.
It has also served as the foundation
for the interactive clustering model of
Balcan and Blum~\cite{balcan-blum:2008:split-merge}
and Awasthi et al.~\cite{awasthi-zadeh:2010:supervised-clustering,%
awasthi-balcan-voevodski:2017:local-algorithm-journal}.

In the \emph{interactive learning framework},
there is an unknown ground truth \structure $s^*$ to be learned.
In each round, the learning algorithm proposes
a \structure $s$ to the user.
In response, with probability $p > \half$,
the user provides correct feedback.
In the remaining case (i.e., with probability $1 - p$),
the feedback is \emph{arbitrary}; in particular,
it could be arbitrarily and deliberately misleading.

Correct feedback is of the following form:
if $s = s^*$, then the algorithm is told this fact
in the form of a user response of $s$.
Otherwise, the user reveals a \structure $s' \neq s$
that is ``more similar'' to $s^*$ than $s$ was.
The exact nature of ``more similar,'' as well as the possibly
restricted set of suggestions $s'$ that the user can propose,
depend on the application domain.
Indeed, the strength of our proposed framework is that
it provides strong query complexity guarantees
under minimal assumptions about the nature of the feedback;
to employ the framework, one merely has to verify that the
the following assumption holds.

\begin{definition}[Graph Model for Feedback]
\label{def:shortest-path}
Define a weighted graph $G$
(directed or undirected) that contains
one node for each \structure $s \in \AllStructures$,
and an edge $(s, s')$ with arbitrary positive
edge \length $\Weight[(s,s')] > 0$ if
the user is allowed to propose $s'$ in response to $s$.
(Choosing the \lengths of edges
is an important part of using the framework.)
$G$ may contain additional edges
not corresponding to any user feedback.
The key property that $G$ must satisfy is the following:
(*) If the algorithm proposes $s$ and the ground truth
is $s^* \neq s$,
then every correct user feedback $s'$
lies on a shortest path from $s$ to $s^*$ in $G$
with respect to the \lengths \Weight[e].
If there are multiple candidate nodes $s'$,
then there is no guarantee on which one
the algorithm will be given by the user.
\end{definition}

\subsection{Algorithm and Guarantees}

Our algorithms are direct reformulations
and slight generalizations of algorithms recently proposed by
Emamjomeh-Zadeh et al.~\cite{2016:binary-search},
which itself was a significant generalization of the natural
``Halving Algorithm'' for learning a classifier
(e.g., \cite{littlestone:1988:online-learning}).
They studied the search problem as an abstract problem
they termed ``Binary Search in Graphs,''
without discussing any applications.
Our main contribution here is
the application of the abstract search problem
to a large variety of interactive learning problems,
and a framework that makes such applications easy.
We begin with the simplest case $p = 1$,
i.e., when the algorithm only receives correct feedback.

Algorithm~\ref{algo:no-errors} gives
essentially best-possible general guarantees
\cite{2016:binary-search}.
To state the algorithm and its guarantees,
we need the notion of an approximate median node
of the graph $G$. First, we denote by
\begin{align*}
  \Reach{s}{s'} :=
   \begin{cases}
      \Set{s} & \text{ if } s' = s \\                    
      \SpecSet{\hat{s}}{s' \text{ lies on a shortest path from }
      s \text{ to } \hat{s}}
      & \text{ if } s' \neq s
   \end{cases}
\end{align*}
the set of all \structures $\hat{s}$ that are consistent
with a user feedback of $s'$ to a \structure $s$.
In anticipation of the noisy case,
we allow \structures to be weighted%
\footnote{Edge \lengths are
part of the definition of the graph,
but node \weightsV will be assigned by our algorithm;
they basically correspond to likelihoods.},
and denote the node \weightsV or \emph{likelihoods}
by $\NodeWeight[s] \geq 0$.
If feedback is not noisy (i.e., $p = 1$),
all the non-zero node \weightsV are equal.
For every subset of \structures \ConsistentSet,
we write $\NodeWeight[\ConsistentSet] :=
\sum_{s \in \ConsistentSet} \NodeWeight[s]$
for the total node \weightV
of the \structures in \ConsistentSet.
Now, for every \structure $s$, define
\begin{align*}
\Potential{\NodeWeight}{s}
:= \frac{1}{\NodeWeight[\AllStructures]} \cdot
\max_{s' \neq s, (s, s') \in G}\NodeWeight[\Reach{s}{s'}]
\end{align*}
to be the largest fraction (with respect to node \weightsV)
of \structures that could still be consistent with a worst-case
response $s'$ to a proposed \structure of $s$.
For every subset of \structures \ConsistentSet,
we denote by $\NodeWeight_\ConsistentSet$
the likelihood function
that assigns \weightV $1$
to every node $s \in \ConsistentSet$
and $0$ elsewhere.
For simplicity of notation, we use
$\Potential{\ConsistentSet}{s}$
when the node \weightsV are $\NodeWeight_\ConsistentSet$.

The simple key insight of \cite{2016:binary-search}
can be summarized and reformulated as the following proposition:

\begin{proposition}[\cite{2016:binary-search},
Proofs of Theorems 3 and 14]
\label{prop:approximate-median-exists}
Let $G$ be a (weighted) directed graph in which each edge $e$
with \length \Weight[e] is part of a cycle of total edge \length
at most $c \cdot \Weight[e]$.
Then, for every node \weightV function \NodeWeight,
there exists a \structure $s$ such that 
$\Potential{\NodeWeight}{s} \leq \frac{c-1}{c}$.

When $G$ is undirected (and hence $c = 2$), 
for every node \weightV function \NodeWeight,
there exists an $s$ such that 
$\Potential{\NodeWeight}{s} \leq \half$.
\end{proposition}

\begin{proof}
For any pair $s, s'$ of \structures,
let $d(s,s')$ denote their distance in $G$
with respect to the edge \lengths \Weight.
Fix a node \weightV function \NodeWeight.
Define
$\MedianPotential{\NodeWeight}{s} :=
\sum_{s' \in \AllStructures} d(s, s') \NodeWeight[s']$
as the total weighted distance from $s$
to every \structure $s'$.
Let $\hat{s}$ be a \structure
that minimizes \MedianPotential{\NodeWeight}{s}.
We prove that $\hat{s}$ satisfies the claim of 
Proposition~\ref{prop:approximate-median-exists}.

Let $e = (\hat{s}, s)$ be an edge in $G$.
By definition, $s$ lies on a shortest path
from $\hat{s}$ to every \structure
$s' \in \Reach{\hat{s}}{s}$.
Therefore, for every $s' \in \Reach{\hat{s}}{s}$,
we have $d(s, s') = d(\hat{s}, s') - \Weight[e]$.
On the other hand, $e$ belongs to a cycle
of total \length at most $c \cdot \Weight[e]$,
so there is a path of total \length
no more than $(c - 1) \cdot \Weight[e]$ from $s$ to $\hat{s}$.
Thus, for every $s' \notin \Reach{\hat{s}}{s}$,
we have $d(s, s') \leq d(\hat{s}, s') + (c - 1) \cdot \Weight[e]$.
We can now bound \MedianPotential{\NodeWeight}{s}
from above:

\begin{align*}
\MedianPotential{\NodeWeight}{s}
& = \sum_{s' \in \AllStructures} \NodeWeight[s'] \cdot d(s, s') \\
& \leq \sum_{s' \in \Reach{\hat{s}}{s}}
\NodeWeight[s'] \cdot (d(\hat{s}, s') - \Weight[e]) 
  + \sum_{s' \notin \Reach{\hat{s}}{s}}
\NodeWeight[s'] \cdot (d(s, s') + (c - 1) \cdot \Weight[e]) \\
& = \MedianPotential{\NodeWeight}{\hat{s}} - \Weight[e] \cdot
\left(
\sum_{s' \in \Reach{\hat{s}}{s}} \NodeWeight[s']
- (c - 1) \cdot \sum_{s' \notin \Reach{\hat{s}}{s}} \NodeWeight[s']
\right).
\end{align*}

By definition of $\hat{s}$, we have
$\MedianPotential{\NodeWeight}{s} \geq
\MedianPotential{\NodeWeight}{\hat{s}}$,
so after dividing by $\Weight[e] > 0$, we get that
$\sum_{s' \in \Reach{\hat{s}}{s}} \NodeWeight[s']
- (c - 1) \cdot \sum_{s' \notin\Reach{\hat{s}}{s}}  \NodeWeight[s']$
must be non-positive, implying that 
$\sum_{s' \in \Reach{\hat{s}}{s}} \NodeWeight[s'] \leq
(c - 1) \cdot \sum_{s' \notin \Reach{\hat{s}}{s}} \NodeWeight[s']$.
This completes the proof.
\end{proof}

In Algorithm~\ref{algo:no-errors},
we always have uniform node \weightV for all the \structures
which are consistent with all the feedback received so far,
and node \weightV $0$ for \structures that are inconsistent
with at least one response.
Prior knowledge about candidates for $s^*$
can be incorporated by providing the algorithm with the input
$\InitSet \ni s^*$ to focus its search on;
in the absence of prior knowledge, the algorithm can be given
$\InitSet = \AllStructures$.

\InsertAlgorithm{Learning a \structure without Feedback Errors
$\normalfont{(\InitSet)}$}%
{algo:no-errors}{
\STATE{$\ConsistentSet \AlgAssign \InitSet$}.
\WHILE{$\SetCard{\ConsistentSet} > 1$}
\STATE{Let $s$ be a \structure
with a ``small'' value of \Potential{\ConsistentSet}{s}. \label{line:potential}}
\STATE{Let $s'$ be the user's feedback \structure.\label{line:response}} 
\STATE{Set $\ConsistentSet \AlgAssign \ConsistentSet \cap \Reach{s}{s'}$.}\label{line:update}
\ENDWHILE
\RETURN{the only remaining \structure in \ConsistentSet.}
}

Line~\ref{line:potential} is underspecified as ``small.''
Typically, an algorithm would choose the $s$ with smallest
\Potential{\ConsistentSet}{s}.
But computational efficiency constraints or other restrictions
(see Sections~\ref{sec:constraints} and \ref{sec:classification})
may preclude this choice and force the algorithm to choose a suboptimal $s$.
The guarantee of Algorithm~\ref{algo:no-errors} is summarized by the
following Theorem~\ref{thm:no-errors}.
It is a straightforward generalization of
Theorems 3 and 14 from
\cite{2016:binary-search}, but for completeness,
we give a self-contained proof.

\begin{theorem} \label{thm:no-errors}
Let $N_0 = \SetCard{\InitSet}$ be the
number of initial candidate \structures.
If each \structure $s$ chosen in Line~\ref{line:potential}
of Algorithm~\ref{algo:no-errors} has
$\Potential{\ConsistentSet}{s} \leq \beta$,
then Algorithm~\ref{algo:no-errors} finds $s^*$ using at most
$\log_{1/\beta} N_0$ queries.
\end{theorem}

\begin{corollary} \label{cor:no-errors:half}
When $G$ is undirected and the optimal $s$ is used in each iteration,
$\beta = \half$ and Algorithm~\ref{algo:no-errors} finds $s^*$
using at most $\log_2 N_0$ queries.
\end{corollary}

\begin{extraproof}{Theorem~\ref{thm:no-errors}}
Let \ConsistentSet be the set of \structures
that are consistent with all the query responses so far.
(Initially, $\ConsistentSet = \InitSet$.)
Let $s$ be a \structure with
$\Potential{\ConsistentSet}{s} \leq \beta$.
If the algorithm proposes $s$ to the user,
the user's feedback will be consistent with
at most a $\beta$ fraction of the \structures in \ConsistentSet.
Note that in Algorithm~\ref{algo:no-errors},
$\NodeWeight(s') = 1$ for every $s' \in \ConsistentSet$.
Given that the set of consistent \structures
shrinks by at least a factor of $\beta$ in each round,
it takes no more than $\log_{1/\beta} \SetCard{\InitSet}$ rounds
to get it down to a single \structure
(which must then be the target).
\end{extraproof}

Next, we present the algorithm in the case of
probabilistically incorrect feedback.
It is a close adaptation of an algorithm by
Emamjomeh-Zadeh et al.~\cite{2016:binary-search},
which resembles a multiplicative-weights update algorithm.
It keeps track of node weights $\NodeWeight(s)$ for each
model, which now exactly correspond to likelihoods of the observed
responses, given that $s$ is the target.
Hence, instead of the generic name ``node weight,'' we will refer to
them as likelihoods.

In order to achieve an information-theoretically optimal dependence on $p$,
\cite{2016:binary-search} run the algorithm in several stages,
removing \emph{very likely} nodes from consideration
for later inspection.
Here, however, we slightly modify (in fact, simplify)
the algorithm by not removing the likely nodes during the execution.
This modification is crucial to make the algorithm
efficiently implementable using a sampling oracle,
as discussed in Section ~\ref{sec:constraints}.
The key ``multiplicative weights'' loop is reformulated and slightly
generalized by Algorithm~\ref{algo:mult-weights}.

\InsertAlgorithm{\MultiWeights ($\normalfont{\InitSet, \TOTQ}$)}%
{algo:mult-weights}{
  \STATE{Set $\NodeWeight[v] \AlgAssign 1$
for all \structures $s \in \InitSet$\\
and $\NodeWeight[s] \AlgAssign 0$ for all \structures $s \notin \InitSet$.}
\STATE{Set $\Marked \AlgAssign \emptyset$.}
\FOR{$\TOTQ + 1$ iterations}
	\IF{there exists a \structure $s$ with
	$\NodeWeight[s] \geq \half \NodeWeight[\InitSet]$}
		\STATE{Mark $s$, by setting $\Marked \AlgAssign \Marked \cup \Set{s}$.}
	\ENDIF
	\STATE{Let $s$ be a \structure with ``small'' \Potential{\NodeWeight}{s}.
	\label{line:errors-potential}}
	\STATE{Query \structure $s$, receiving feedback $s'$.}
	\FORALL{\structures $\hat{s} \in \InitSet$}
		\IF{$\hat{s} \in \Reach{s}{s'}$} 
			\STATE{$\NodeWeight[\hat{s}] \AlgAssign
			p \cdot \NodeWeight[\hat{s}]$.}
		\ELSE
			\STATE{$\NodeWeight[\hat{s}] \AlgAssign
			(1- p) \cdot \NodeWeight[\hat{s}]$.}
		\ENDIF
	\ENDFOR
\ENDFOR
\RETURN{$\Marked$}
}

Algorithm~\ref{algo:mult-weights} is invoked several times
in Algorithm~\ref{algo:errors}.
Algorithm~\ref{algo:errors} is given
an initial candidate \structure set \InitSet,
and also a \emph{threshold} $0 < \AlgThreshold < 1$
(to be specified later),
as well as a target error probability \ErrProb.
The algorithm must output $s^*$ (the target \structure)
with probability at least $1 - \ErrProb$,
assuming that $s^* \in \InitSet$.
The exact constants used are chosen to achieve an optimal dependency
of the number of queries on $p$ in the worst case over graphs
\cite{2016:binary-search}.
Here and below, $\Entr{p} = - p \log p - (1 - p) \log (1 - p)$
denotes the entropy.

\InsertAlgorithm{Online Learning of a \structure with Imperfect Feedback
$\normalfont{(\InitSet, \AlgThreshold, \ErrProb)}$}{algo:errors}{
\STATE{$\ErrProb' \AlgAssign \ErrProb/5$.}
\STATE{Fix $\Tolerance_1 = \max\Set{
\sqrt{\frac{1}{\log \log \SetCard{\InitSet}}},
\frac{\log (1/\AlgThreshold) - \Entr{p}}{2 \log(p/(1 - p))}}$.}
\label{line:lambda1}
\STATE{$\TOTQ_1 \AlgAssign
\max\Set{\frac{\log \SetCard{\InitSet}}{\log (1/\AlgThreshold)
- \Entr{p} - \Tolerance_1 \log(p/(1 - p))}
+ \frac{\log(1/\AlgThreshold)}{\log(1/\ErrProb')},
\frac{\ln(1 / \ErrProb')}{\Tolerance_1 ^ 2}} + 1$.}
\STATE{$S_1 \AlgAssign \MultiWeights(V, \TOTQ_1)$.}
\STATE{Fix $\Tolerance_2 =
\frac{\log (1/\AlgThreshold) - \Entr{p}}{2 \log(p/(1 - p))}$.}
\label{line:lambda2}
\STATE{$\TOTQ_2 \AlgAssign \max\Set{
\frac{\log(\SetCard{S_1})}{\log(1/\AlgThreshold) - \Entr{p}
- \Tolerance_2 \log(p/(1 - p))}
+ \frac{\log(1/\AlgThreshold)}{\log(1/\ErrProb')},
\frac{\ln(1 / \ErrProb')}{\Tolerance_2 ^ 2}} + 1$.}
\STATE{$S_2 \AlgAssign \MultiWeights(S_1, \TOTQ_2)$.}
\FORALL{$s \in S_2$}
	\STATE{Query $s$ repeatedly
	$\frac{2 \ln(\SetCard{S_2}/\ErrProb')}{(2p - 1)^2}$ times.}
	\IF{$s$ is returned as the correct \structure
	for at least half of these queries}
		\RETURN{$s$.}
	\ENDIF
\ENDFOR
\RETURN{failure.}
}

The performance of Algorithm~\ref{algo:errors} is summarized
in Theorem~\ref{thm:errors},
which generalizes the results of \cite{2016:binary-search}
to arbitrary values of $\beta$.

\begin{theorem} \label{thm:errors}
Let $\beta \in [\half, 1)$,
define $\AlgThreshold = \beta p + (1 - \beta)(1 - p)$,
and let $N_0 = \SetCard{\InitSet}$.
Assume that $\log (1/\AlgThreshold) > \Entr{p}$
where $\Entr{p} = - p \log p - (1 - p) \log(1- p)$
denotes the entropy.
(When $\beta = \half$,
this holds for every $p > \half$.)

If each \structure $s$ chosen in Line~\ref{line:errors-potential}
of Algorithm~\ref{algo:mult-weights} has
$\Potential{\NodeWeight}{s} \leq \beta$,
then with probability at least $1 - \ErrProb$,
Algorithm~\ref{algo:errors} finds $s^*$ using at most
$\frac{(1 - \ErrProb)}{\log (1/\AlgThreshold) - \Entr{p}} \log N_0
+ o(\log N_0) + O(\log^2(1/\ErrProb))$
queries in expectation.
\end{theorem}

\begin{corollary} \label{cor:errors:half}
When the graph $G$ is undirected and the optimal $s$ is used in
each iteration, then with probability at least $1 - \ErrProb$,
Algorithm~\ref{algo:errors} finds $s^*$ using at most
$\frac{(1 - \ErrProb)}{1 - \Entr{p}} \log_2 N_0
+ o(\log N_0) + O(\log^2(1/\ErrProb))$
queries in expectation.
\end{corollary}

\begin{extraproof}{Theorem~\ref{thm:errors}}
The proof of this theorem is fairly similar
to the proof of Lemma 7 in \cite{2016:binary-search}.
Here, we give a high-level outline of that proof,
and specify the changes that need to be made for our
slightly more general setting,
as well as the modification we made
in Algorithm~\ref{algo:mult-weights}.

Let $s$ be the \structure in Line~\ref{line:errors-potential}
of Algorithm~\ref{algo:mult-weights}. If it has
$\NodeWeight[s] \leq \half \NodeWeight[\InitSet]$,
whatever the query response is,
at most a $\beta$ fraction of the total likelihood
is multiplied by $p$, while the rest is multiplied by $1 - p$.
Therefore, the total likelihood
decreases by at least a factor of
$\AlgThreshold = \beta p + (1 - \beta)(1 - p)$.
Furthermore, if $\NodeWeight[s] > \half \NodeWeight[\InitSet]$
and $s$ is not the target, we show that in expectation,
the total likelihood decreases by at least a factor \half.
This is because with probability at most $1 - p$,
the query response claims that $s$ is the target
in which case \NodeWeight[s] is multiplied by $p$,
while all other nodes' likelihoods are multiplied by $1 - p$.
In the remaining cases, i.e., if the response is correct,
or it is incorrect but pointing to a neighbor of $s$,
\NodeWeight[s] is multiplied by $1 - p$.
Hence, in expectation, the total likelihood decreases
by a factor of
\begin{align*}
(1 - p) \cdot \Big( p \cdot \NodeWeight[s]
+ (1 - p)\cdot (\NodeWeight[\InitSet] - \NodeWeight[s]) \Big)
+ p \cdot \Big( (1 - p)\NodeWeight[s]
+ p \cdot (\NodeWeight[\InitSet] - \NodeWeight[s]) \Big)\\
\leq \half \NodeWeight[\InitSet]
\leq \AlgThreshold \NodeWeight[\InitSet].
\end{align*}

To summarize, in expectation, the total likelihood
decreases by at least a factor \AlgThreshold in both cases.
Initially, in Algorithm~\ref{algo:mult-weights},
the total likelihood of all \structures is $\SetCard{\InitSet}$,
and in each iteration, the total likelihood
decreases at least by a factor of \AlgThreshold, in expectation.
By induction, after \TOTQ rounds, 
the expected total likelihood
is bounded by $\SetCard{\InitSet} \cdot \AlgThreshold^\TOTQ$.
On the other hand, in expectation, a $p$ fraction of the
query responses is correct.
Using tail bounds, for a carefully chosen value of $\Tolerance$
(as in Algorithm~\ref{algo:errors},
Lines~\ref{line:lambda1} and \ref{line:lambda2}),
with ``high enough'' probability,
a $p - \Tolerance$ fraction of the responses are correct.
Whenever a $p - \Tolerance$ fraction of the responses is correct,
the likelihood \NodeWeight[s^*] of the target \structure $s^*$
is at least
$(p ^ {p - \Tolerance} \cdot (1 - p) ^ {1 - p + \Tolerance})^\TOTQ$.

Because the statement of the theorem assumed that
$\log (1/\AlgThreshold) \; > \; H(p) = -p \log p - (1 - p) \log(1 - p)$,
or, equivalently, $\AlgThreshold < p^p (1 - p)^{1 - p}$,
\NodeWeight[s^*] decreases exponentially slower
than the total likelihood of all the \structures;
hence, $s^*$ must eventually get marked by
Algorithm~\ref{algo:mult-weights} with high probability.

In fact, $\TOTQ_1$ and $\TOTQ_2$ are chosen such that this marking of
$s^*$ happens with sufficiently high probability within that many rounds.
The precise calculations are quite similar to those in
\cite{2016:binary-search}, but slightly more messy.
\cite{2016:binary-search} showed that if the first time that
Algorithm~\ref{algo:mult-weights} 
is invoked by Algorithm~\ref{algo:errors},
it is executed for (roughly)
$\frac{\log N}{1 - \Entr{p}}$ iterations,
then with high probability, $s^*$ is marked.
This analysis can be straightforwardly generalized to our setting,
when running for $\frac{\log N}{\log (1/\AlgThreshold) - \Entr{p}}$ rounds.
The proof for the other stages of Algorithm~\ref{algo:errors}
remains almost the same.
\end{extraproof}

\subsection{Computational Considerations and Sampling}
\label{sec:constraints}

Corollaries~\ref{cor:no-errors:half} and \ref{cor:errors:half}
require the algorithm to find a \structure $s$
with small \Potential{\NodeWeight}{s} in each iteration.
In most learning applications, the number $N$ of candidate \structures
is exponential in a natural problem parameter $n$,
such as the number of sample points (classification),
or the number of items to rank or cluster.
If computational efficiency is a concern,
this precludes explicitly keeping track of the set \ConsistentSet
or the \weightsV \NodeWeight[s].
It also rules out determining the \structure $s$ to query
by exhaustive search over all \structures
that have not yet been eliminated.

In some cases, these difficulties can be circumvented
by exploiting problem-specific structure.
A more general approach relies on Monte Carlo techniques.
We show that the ability to sample \structures $s$
with probability (approximately) proportional to \NodeWeight[s]
(or approximately uniformly from \ConsistentSet
in the case of Algorithm~\ref{algo:no-errors})
is sufficient to essentially achieve the results of
Corollaries~\ref{cor:no-errors:half} and \ref{cor:errors:half}
with a computationally efficient algorithm.
Notice that in both Algorithms~\ref{algo:no-errors}
and \ref{algo:mult-weights},
\NodeWeight[s] is completely determined by
all the query responses the algorithm has seen so far
and the probability $p$.

\begin{theorem} \label{thm:main-sampling-theorem}
Let $n$ be a natural measure of the input size
and assume that $\log N$ is polynomial in $n$.
Assume that $G = (V, E)$ is undirected%
\footnote{It is actually sufficient that
for every node \weightV function
$\NodeWeight : V \rightarrow \R^{+}$,
there exists a \structure $s$
with $\Potential{\NodeWeight}{s} \leq \half$.},
all edge \lengths are integers,
and the maximum degree and diameter
(both with respect to the edge \lengths)
are bounded by $\poly(n)$.
Also assume w.l.o.g.~that \NodeWeight is normalized
to be a distribution over the nodes%
\footnote{For Algorithm~\ref{algo:no-errors},
\NodeWeight is uniform over all \structures
consistent with all feedback up to that point.} 
(i.e., \NodeWeight[\AllStructures] = 1).

Let $0 \leq \TVDis < \frac{1}{4}$ be a constant,
and assume that there is an oracle that
--- given a set of query responses ---
runs in polynomial time in $n$
and returns a \structure $s$ drawn from a distribution \NodeWeightP
with $\tvd{\NodeWeight}{\NodeWeightP} \leq \TVDis$.
Also assume that there is a polynomial-time algorithm that,
given a \structure $s$, decides whether or not
$s$ is consistent with every given query response or not.

Then, for every $\epsilon > 0$,
in time $\poly(n, \frac{1}{\epsilon})$, an algorithm can find a
\structure $s$ with $\Potential{\NodeWeight}{s} \leq
\half + 2 \TVDis + \epsilon$,
with high probability.
Therefore, Algorithms~\ref{algo:no-errors} and \ref{algo:errors}
with $\beta = \half + 2 \TVDis + \epsilon$ (for a sufficiently small $\epsilon$)
can be implemented to run in time $\poly(n, \frac{1}{\epsilon})$.
\end{theorem}

\begin{proof}
The high-level idea is to use a simple local search algorithm
to find a \structure $s$ with small \Potential{\NodeWeight}{s}.
In order to execute the updating step,
we need to estimate \NodeWeight[\Reach{s}{s'}]
for all neighbors $s'$ of $s$ in $G$.
The key insight here is that \NodeWeight[\Reach{s}{s'}]
is exactly the probability that a node drawn from \NodeWeight
is consistent with the feedback $s'$ when $s$ is queried.
To get a sharp enough estimate of \NodeWeight[\Reach{s}{s'}],
in each iteration, enough samples are drawn to ensure that
tail bounds kick in and provide high-probability guarantees.
The high-level algorithm is given as
Algorithm~\ref{algo:sampling-median};
details --- in particular on Line~\ref{line:empirical} ---
are provided below.

\InsertAlgorithm{Finding a good approximate median
$\normalfont{(\TVDis, \epsilon)}$}%
{algo:sampling-median}{
\STATE{Let $\epsilon' = \epsilon/3$.}
\STATE{Let $s$ be an arbitrary \structure}.
\LOOP
\FOR{each edge $e = (s, s')$}
\STATE{Let \EventProb{s}{s'} be the empirical estimation
of \NodeWeight[\Reach{s}{s'}]}. \label{line:empirical}
\ENDFOR 
\IF{$\EventProb{s}{s'} \leq \half + \TVDis + 2 \epsilon'$
for every edge $e = (s, s')$}
\RETURN{$s$}.
\ELSE
\STATE{Let $e = (s, s')$ be an edge out of $s$ with
$\EventProb{s}{s'} > \half + \TVDis + 2 \epsilon'$.}
\STATE{Set $s \AlgAssign s'$.} \label{line:switch-model}
\ENDIF
\ENDLOOP
}

In Line~\ref{line:empirical} of Algorithm~\ref{algo:sampling-median},
$\EventProb{s}{s'}$ is estimated as the fraction of samples $\hat{s}$
drawn from \NodeWeightP that are in \Reach{s}{s'}.
Below, for a particular desired high-probability guarantee,
we choose the number of samples to guarantee that
\begin{equation} \label{eq:sampling-guarantee}
|\EventProb{s}{s'} - \NodeWeight[\Reach{s}{s'}]|
\leq \TVDis + \epsilon'
\end{equation}
with high probability.
For most of the proof, we will assume that all of these
high-probability events occurred;
at the end, we will calculate the probability of this happening.

Whenever an edge $e = (s, s')$ has
$\EventProb{s}{s'} \leq \half + \TVDis + 2 \epsilon'$,
we can bound
$\NodeWeight[\Reach{s}{s'}] \leq \half + 2 \TVDis + 3 \epsilon'
= \half + 2 \TVDis + \epsilon$.
Thus, whenever Algorithm~\ref{algo:sampling-median}
returns a \structure $s$,
the \structure satisfies
$\Potential{\NodeWeight}{s} \leq \half + 2 \TVDis + \epsilon$.

It remains to show that the algorithm terminates,
and in a polynomial number of iterations.
Define
$\MedianPotential{\NodeWeight}{s} = \sum_{s'} \NodeWeight[s'] d(s, s')$
to be the weighted distance
(with respect to the edge \lengths \Weight[e])
from $s$ to all other \structures.
When Algorithm~\ref{algo:sampling-median} switches from $s$ to $s'$
in Line~\ref{line:switch-model},
by Inequality~\eqref{eq:sampling-guarantee}, we have
$\NodeWeight[\Reach{s}{s'}] \geq
\half + \TVDis + 2 \epsilon' - (\TVDis + \epsilon')
= \half + \epsilon'$,
so $\MedianPotential{\NodeWeight}{s'}
\leq \MedianPotential{\NodeWeight}{s} - 2 \epsilon' \cdot \Weight[e]$.
Because the value of \MedianPotential{\NodeWeight}{s}
for all $s$ is bounded by the diameter of $G$,
which was assumed to be bounded by $\poly(n)$,
and the edge \lengths are all positive integers,
Algorithm~\ref{algo:sampling-median} terminates after
$\poly(n, 1/\epsilon)$ steps.

The only remaining part is to compute
how many samples from the oracle are required
to guarantee Inequality~\eqref{eq:sampling-guarantee}
with high enough probability.
Drawing \NumSample samples,
by standard tail bounds, Inequality~\eqref{eq:sampling-guarantee}
fails with probability no more than
$e^{-\Omega(\epsilon^2 \NumSample)}$.
Because all degrees are bounded by $\poly(n)$,
Line~\ref{line:empirical} of Algorithm~\ref{algo:sampling-median} is
executed only $\poly(n, 1/\epsilon)$ times,
so it suffices to draw
$\NumSample = \poly(\frac{1}{\epsilon}, \log n)$
samples in each iteration
in order to apply a union bound over all iterations of the algorithm.
\end{proof}


\section{Application I: Learning a Ranking} \label{sec:sorting}

As a first application, we consider the task of learning the
correct order of $n$ elements with supervision in the form of
equivalence queries.
This task is directly motivated by learning a user's preference
over web search results
(e.g., \cite{joachims:2002:search-engine-clickthrough,
radlinski-joachims:2005:query-chain}),
restaurant or movie orders
(e.g., \cite{crammer-singer:2002:pranking}),
or many other types of entities.
Using pairwise \emph{active} queries
(``Do you think that A should be ranked ahead of B?''),
a learning algorithm could of course simulate standard
$O(n \log n)$ sorting algorithms;
this number of queries is necessary and sufficient.
However, when using equivalence queries, the user must be
presented with a complete ordering
(i.e., a permutation \PER of the $n$ elements),
and the feedback will be a \emph{mistake} in the proposed permutation.
Here, we propose interactive algorithms
for learning the correct ranking
without additional information or assumptions.%
\footnote{For example,
\cite{joachims:2002:search-engine-clickthrough,%
radlinski-joachims:2005:query-chain,%
crammer-singer:2002:pranking}
map items to feature vectors
and assume linearity of the target function(s).} 
We first describe results for a setting
with simple feedback in the
form of adjacent transpositions; 
we then show a generalization
to more realistic feedback as one is
wont to receive in applications such as search engines.

\subsection{Adjacent Transpositions}
\label{sec:sorting:bubble}

We first consider ``\Bsort'' feedback of the following form:
the user specifies that elements $i$ and $i + 1$ in the proposed
permutation \PER are in the wrong relative order.
An obvious correction for an algorithm would be to swap
the two elements, and leave the rest of \PER intact.
This algorithm would exactly implement \Bsort,
and thus require $\Theta(n^2)$ equivalence queries.
Our general framework allows us to easily obtain
an algorithm with $O(n \log n)$ equivalence queries instead.
We define the undirected and unweighted graph \BSG as follows:

\begin{itemize}

\item \BSG contains $N = n!$ nodes,
one for each permutation \PER of the $n$ elements;

\item it contains an edge between \PER and \PERP if and only if
\PERP can be obtained from \PER by swapping two adjacent elements.

\end{itemize}

\begin{lemma} \label{lem:bubblesort-graph}
\BSG satisfies Definition~\ref{def:shortest-path}
with respect to \Bsort feedback.
\end{lemma}

\begin{proof}
By definition, there is one node for each permutation,
and the edges of \BSG exactly capture the possible feedback
provided in response to the query.
For two permutations $\PER, \sorted$, let
$\tau(\PER, \sorted) = \SetCard{\SpecSet{(i,j)}{i \text{ precedes } j
\text{ in } \PER \text{ and follows } j \text{ in } \sorted}}$
denote the Kendall $\tau$ distance.
Since an adjacent transposition can fix the ordering of
exactly one pair $(i,j)$, the distance between $\PER, \sorted$
in \BSG is exactly equal to $\tau(\PER, \sorted)$.
If the user proposes \PERP in response to \PER,
then \PERP must fix exactly one inversion between \PER and \sorted,
so the distance between \PERP and \sorted is
$\tau(\PERP, \sorted) = \tau(\PER, \sorted) - 1$.
Thus, \PERP lies on a shortest path from \PER to \sorted.
\end{proof}

Hence, applying Corollary~\ref{cor:no-errors:half}
and Theorem~\ref{thm:errors},
we immediately obtain the existence of learning algorithms
with the following properties:

\begin{corollary} \label{cor:bubblesort-algo}
Assume that in response to each equivalence query
on a permutation \PER,
the user responds with an adjacent transposition
(or states that the proposed permutation \PER is correct).
\begin{enumerate}
\item If all query responses are correct,
then the target ordering
can be learned by an interactive algorithm using at most
$\log N = \log n! \leq n \log n$ equivalence queries.
\item If query responses are correct with probability $p > \half$,
the target ordering can be learned by an interactive algorithm
with probability at least $1 - \ErrProb$ using at most
$\frac{(1 - \ErrProb)}{1 - \Entr{p}} n \log n
+ o(n \log n) + O(\log^2(1/\ErrProb))$
equivalence queries in expectation.
\end{enumerate}
\end{corollary}
  
Up to constants, the bound of
Corollary~\ref{cor:bubblesort-algo}
is optimal:
Theorem~\ref{thm:sorting-lower-bound} shows that
$\Omega(n \log n)$ equivalence queries
are necessary in the worst case.
Notice that Theorem~\ref{thm:sorting-lower-bound}
does not immediately follow
from the classical lower bound for sorting
with pairwise comparisons:
while the result of a pairwise comparison
always reveals one bit,
there are $n - 1$ different possible responses
to an equivalence query, 
so up to $O(\log n)$ bits might be revealed.
For this reason,
the proof of Theorem~\ref{thm:sorting-lower-bound}
explicitly constructs an adaptive adversary,
and does not rely on a simple counting argument.

\begin{theorem} \label{thm:sorting-lower-bound}
With adversarial responses,
any interactive ranking algorithm
can be forced to ask $\Omega(n \log n)$ equivalence queries.
This is true even if the true ordering
is chosen uniformly at random,
and only the query responses are adversarial.
\end{theorem}

\begin{proof}
Let \AlgSym be an arbitrary interactive algorithm
which learns the underlying order.
In each round, \AlgSym proposes a permutation \PER
and receives feedback in the form of an adjacent transposition.
We start with the case when
both the correct ordering and the responses are adversarial.
We define an adaptive adversary \AdvSym which forces \AlgSym to take
$\Omega(n \log n)$ rounds before it finds the correct ordering \sorted.

The adversary \AdvSym gradually marks elements as \PosLabel or
\NegLabel, trying to delay labeling for as long as possible.
Initially, all $n$ elements are unmarked.
When \AlgSym proposes a permutation \PER,
\AdvSym looks for two consecutive elements
$i = \PER_k, j = \PER_{k + 1}$ in \PER which are both unmarked.
If such a pair exists, \AdvSym returns $(k,k + 1)$ as feedback,
i.e., it tells the algorithm that $i$ and $j$ are in the wrong order.
Then, it marks $i$ with \PosLabel and $j$ with \NegLabel.
Because \AdvSym marks exactly two new elements in each round,
and any permutation with fewer than
$\Floor{n/2}$ marked elements must have
two consecutive marked elements, 
\AdvSym can keep following this strategy for at least
$\Floor{\frac{n}{4}}$ rounds.

At this point, the algorithm starts its second stage.
It partitions the elements into two sets:
$\ListSym_{-}$ and $\ListSym_{+}$.
The partition satisfies the following conditions:
\begin{itemize}
\item All elements marked with \NegLabel are in $\ListSym_{-}$,
and all elements marked with \PosLabel are in $\ListSym_{+}$.
\item Elements are partitioned as evenly as possible:
$\Size{\ListSym_{-}} \geq \Floor{\frac{n}{2}}$
and
$\Size{\ListSym_{+}} \geq \Floor{\frac{n}{2}}$.
\end{itemize}
Notice that such a partition always exists,
because the number of elements marked \PosLabel and \NegLabel
are the same.
\AdvSym commits to the fact that in \sorted,
all elements of  $\ListSym_{-}$ will precede all elements of
$\ListSym_{+}$, but leaves \sorted unspecified beyond that.
Notice that \AlgSym has not received any feedback on the relative
ordering of any pair of elements in $\ListSym_{-}$ or any pair in 
$\ListSym_{+}$.
In the subsequent rounds, \AdvSym responds as follows:
\begin{itemize}
\item If \AlgSym proposes a permutation \PER in which an element
$j \in \ListSym_{-}$ appears after an element $i \in \ListSym_{+}$,
then \PER must contain an \emph{adjacent} such pair $(i,j)$.
In that case, \AdvSym provides the feedback that $(i, j)$
are out of order.
In fact, this provides \AlgSym with no new information beyond
$\ListSym_{-}$ and $\ListSym_{+}$,
which it was assumed to already know.
\item Otherwise, in \PER, all elements of $\ListSym_{-}$
must precede  all elements of $\ListSym_{+}$.
In that case, \AdvSym will first recursively follow the same
strategy on $\ListSym_{-}$,
and then recursively follow the same strategy on $\ListSym_{+}$.
\end{itemize}

Let $\QC(n)$ be the minimum number of rounds
that \AlgSym requires against this adversarial strategy \AdvSym.
The description of the adversary strategy implies
the following recurrence:
$\QC(n) \geq \Floor{\frac{n}{4}} + 2 \cdot \QC(\Floor{\frac{n}{2}})$. 
Thus, by the master theorem, $\QC(n) = \Omega(n \log n)$,
proving the theorem.

In the given proof, \AdvSym provides adversarial feedback
and gradually commits to an adversarial permutation.
However, even if the true permutation \sorted 
is chosen uniformly at random ahead of time,
\AdvSym can still follow essentially the same approach
to provide adversarial feedback with respect to \sorted.

Fix \sorted, and consider the first stage of \AdvSym,
where it assigns \NegLabel and \PosLabel to the items.
When \AlgSym proposes a permutation \PER,
every pair $(i,j)$ of adjacent unmarked items in \PER
is, with probability exactly \half,
in the wrong order with respect to \sorted.
If $n'$ elements are already marked,
there are at least $\frac{n - 2n' - 1}{2}$
\emph{disjoint} pairs of adjacent unmarked items,
each of which is in the wrong order
independently with probability exactly \half.
Thus, for each of the first $\frac{n}{8} - 1$ rounds,
the probability that \AdvSym fails to find an unmarked adjacent pair
in the wrong order is less than $2^{-\frac{n}{4}}$.
After $\frac{n}{8} - 1$ rounds, it starts its second stage,
and a similar argument applies.
Because the failure probability is exponentially low, we can take a
union bound over all rounds until $n$ is sufficiently small,
and obtain that with high probability,
\AlgSym requires $\Omega(n \log n)$ queries on a random permutation.
\end{proof}

\subsection{Implicit Feedback from Clicks}
\label{sec:sorting:insertion}

In the context of search engines, it has been argued (e.g., by
\cite{joachims:2002:search-engine-clickthrough,%
radlinski-joachims:2005:query-chain,%
agichtein-brill-dumais-ragno:2006:learning-web-search})
that a user's clicking behavior provides implicit feedback
of a specific form on the ranking.
Specifically, since users will typically read the search results
from first to last,
when a user skips some links that appear earlier in the ranking,
and instead clicks on a link that appears later,
her action suggests that the later link
was more informative or relevant.

Formally, when a user clicks on the element at index $i$,
but did not previously click on any elements at indices
$j, j+1, \ldots, i-1$,
this is interpreted as feedback that element $i$
should precede all of elements $j, j+1, \ldots, i-1$.
Thus, the feedback is akin to an ``\Isort'' move.
(The \Bsort feedback model is the special case in
which $j = i - 1$ always.)

To model this more informative feedback,
the new graph \ISG has more edges,
and the edge \lengths are non-uniform.
It contains the same $N$ nodes (one for each permutation).
For a permutation \PER and indices $1 \leq j < i \leq n$,
\ShiftPer{\PER}{i}{j} denotes the permutation
that is obtained by moving the \Ordinal{i} element in \PER
before the \Ordinal{j} element
(and thus shifting elements $j, j+1, \ldots, i-1$
one position to the right).
In \ISG, for every permutation \PER
and every $1 \leq j < i \leq n$,
there is an undirected edge from \PER to \ShiftPer{\PER}{i}{j}
with \length $i - j$.
Notice that for $i > j + 1$, there is actually no user feedback
corresponding to the edge from \ShiftPer{\PER}{i}{j} to \PER;
however, additional edges are permitted, and
Lemma~\ref{lem:insertionsort-graph}
establishes that \ISG does in fact
satisfy the ``shortest paths'' property.

\begin{lemma} \label{lem:insertionsort-graph}
\ISG satisfies Definition~\ref{def:shortest-path}
with respect to \Isort feedback.
\end{lemma}

\begin{proof}
As in the proof of Lemma~\ref{lem:bubblesort-graph},
there is one node for each permutation,
and for each possible Insertion Sort feedback given
in response to \PER, there is an edge in \ISG
by definition (though \ISG contains additional edges).

As with \BSG, the distance between two permutations $\PER, \sorted$
in \ISG (now with respect to the given edge \lengths)
is still equal to the Kendall $\tau(\PER, \sorted)$.
This is because going from \PER to \ShiftPer{\PER}{i}{j}
transposes exactly $i - j$ pairs of elements, which is the \length
of the corresponding edge,
so no edge with $i > j + 1$ is essential for any shortest path.
  
When the user feedback $\PERP = \ShiftPer{\PER}{i}{j}$
indicates that the \Ordinal{i} element should be placed
ahead of all the elements in positions $j, j + 1, \ldots, i - 1$,
then the \Ordinal{i} element does indeed precede
all of these elements in \sorted.
Thus, $\tau(\PERP, \sorted) = \tau(\PER, \sorted) - (i - j)$,
and the edge $(\PER, \PERP)$ lies on a shortest path
from \PER to \sorted.
Notice that the ``extraneous'' edges of the form
$(\ShiftPer{\PER}{i}{j}, \PER)$ for $i > j + 1$ (of cost $i - j$)
do not affect any shortest paths between pairs
$\PER, \sorted$, as they can be replaced with the corresponding
$i - j$ adjacent transpositions,
without increasing the total \length of the path.  
\end{proof}

As in the case of \BSG,
by applying Corollary~\ref{cor:no-errors:half}
and Theorem~\ref{thm:errors},
we immediately obtain the existence of
interactive learning algorithms
with the same guarantees as those of
Corollary~\ref{cor:bubblesort-algo}.

\begin{corollary} \label{cor:insertionsort-algo}
Assume that in response to each equivalence query, the user responds
with a pair of indices $j < i$ such that element $i$ should precede
all elements $j, j+1, \ldots, i-1$.
\begin{enumerate}
\item If all query responses are correct, then the target
ordering can be learned by an interactive algorithm using at most
$\log N = \log n! \leq n \log n$ equivalence queries.
\item If query responses are correct with probability $p > \half$,
the target ordering can be learned by an interactive algorithm with
probability at least $1 - \ErrProb$ using at most
$\frac{(1 - \ErrProb)}{1 - \Entr{p}} n \log n
+ o(n \log n) + O(\log^2(1/\ErrProb))$
equivalence queries in expectation.
\end{enumerate}
\end{corollary}

\subsection{Computational Considerations}

While Corollaries~\ref{cor:bubblesort-algo} and
\ref{cor:insertionsort-algo}
imply interactive algorithms
using only $O(n \log n)$ equivalence queries,
they do not guarantee that
the internal computations of the algorithms are efficient.
The na\"{\i}ve implementation requires
explicitly keeping track of and
comparing likelihoods on all $N = n!$ nodes.

When $p = 1$, i.e., the algorithm only receives correct feedback,
it can be made computationally efficient
using Theorem~\ref{thm:main-sampling-theorem}.
To apply Theorem~\ref{thm:main-sampling-theorem},
it suffices to show that one can efficiently sample a (nearly)
uniformly random permutation \PER consistent with all feedback
received so far.
Since the feedback is assumed to be correct,
the set of all pairs $(i, j)$ such that the user implied that
element $i$ must precede element $j$ must be acyclic,
and thus must form a partial order.
The sampling problem is thus exactly the problem of
sampling a \emph{linear extension} of a given partial order.

This is a well-known problem, and a beautiful result of
Bubley and Dyer~\cite{bubley-dyer:1999:linear-extensions,%
bubley:2001:randomized-algorithms}
shows that the Karzanov-Khachiyan Markov
Chain~\cite{karzanov-khachiyan:1991:conductance} mixes rapidly.
Huber~\cite{huber:2006:linear-extensions} shows
how to modify the Markov Chain sampling technique
to obtain an exactly (instead of approximately)
uniformly random linear extension of the given partial order.
For the purpose of our interactive learning algorithm,
the sampling results can be summarized as follows:

\begin{theorem}[Huber~\cite{huber:2006:linear-extensions}]
\label{thm:sampling-partial-deterministic}
Given a partial order over $n$ elements,
let \LinearExt be the set of all linear extensions,
i.e., the set of all permutations consistent with the partial order.
There is an algorithm that runs in expected time $O(n^3 \log n)$
and returns a uniformly random sample from \LinearExt.
\end{theorem}

The maximum node degree in \BSG is $n - 1$,
while the maximum node degree in \ISG is $O(n^2)$.
The diameter of both \BSG and \ISG is $O(n^2)$.
Substituting these bounds and the bound from
Theorem~\ref{thm:sampling-partial-deterministic} into
Theorem~\ref{thm:main-sampling-theorem},
we obtain the following corollary:

\begin{corollary} \label{cor:efficient-sorting}
Both under \Bsort feedback and \Isort feedback,
if all feedback is correct,
there is an \emph{efficient} interactive
learning algorithm using at most $\log n! \leq n \log n$
equivalence queries to find the target ordering.
\end{corollary}

The situation is significantly more challenging when feedback
could be incorrect, i.e., when $p < 1$.
In this case, the user's feedback
is not always consistent and may not form a partial order.
In fact, we prove the following hardness result.

\begin{theorem} \label{thm:sampling-hardness-permutations}
There exists a $p$ (depending on $n$) for which the following holds.
Given a set of user responses,
let \NodeWeight[\PER] be the likelihood of \PER given the responses,
and normalized so that $\sum_{\PER} \NodeWeight[\PER] = 1$.
Let $0 < \TVDis < 1$ be any constant.
There is no polynomial-time algorithm to draw a sample
from a distribution \NodeWeightP
with $\tvd{\NodeWeight}{\NodeWeightP} \leq 1 - \TVDis$
unless $\Complexity{RP} = \Complexity{NP}$.
\end{theorem}

\begin{proof}
We prove this theorem using a reduction
from \textsc{Minimum Feedback Arc Set},
a well-known \Hardness{NP}{complete} problem~\cite{karp:1972:reducibility}.
Given a directed graph $G$ and number $k$,
the \textsc{Minimum Feedback Arc Set} problem asks if there is
a set of at most $k$ arcs of $G$ whose removal
will leave the remaining graph acyclic.
This is equivalent to asking if there is a permutation \PER
of the nodes of $G$ such that at most $k$ arcs go
from higher-numbered nodes in \PER to lower-numbered ones.
  
Given \TVDis, a graph $G$ with $n$ nodes and $m$ edges, and $k$,
we define the following sampling problem.
Consider sampling from permutations of $n$ elements,
let $p = 1 - \frac{1}{2(n + 1)!}$,
and let the $m$ user responses be exactly
the (directed) edges of $G$.
  
For any permutation \PER,
let $x_{\PER}$ be the number of queries that \PER agrees with,
and $y_{\PER}$ the number of queries that \PER disagrees with.
Then, for all \PER, $x_{\PER} + y_{\PER} = m$,
and the (unnormalized) likelihood for \PER is
$L(\PER) = p^{x_{\PER}} \cdot (1 - p)^{y_{\PER}}$.
Let $y^* = \min_{\PER} y_{\PER}$, and
let $\Pi^* = \SpecSet{\PER}{y_{\PER} = y^*}$
be the set of all permutations minimizing $y_{\PER}$.
Then, for any permutation $\PER \in \Pi^*$, we have
\[
L(\PER)
\; = \; \left( 1 - \frac{1}{2(n + 1)!} \right)^{m - y^*}
\cdot \left( \frac{1}{2(n + 1)!} \right)^{y^*}
\; =: \; L^*.
\]
On the other hand, for any permutation $\PERP \notin \Pi^*$,
we get that
\begin{align*}
L(\PERP) & \leq \;
\left( 1 - \frac{1}{2(n + 1)!} \right)^{m - y^* - 1}
\cdot \left( \frac{1}{2(n + 1)!} \right)^{y^* + 1}
\\ & = L^* \cdot \frac{1}{2(n + 1)!}
/ \left( 1 - \frac{1}{2(n + 1)!} \right)
\\ & \leq \; \frac{L^*}{(n + 1)!}.
\end{align*}
           
Thus, under the normalized likelihood distribution \NodeWeight,
the total sampling probability of all permutations $\PER \in \Pi^*$
must be

\begin{align*}
\sum_{\PER \in \Pi^*} \NodeWeight[\PER] & = \;
\frac{\sum_{\PER \in \Pi^*}
L(\PER)}{\sum_{\PERP}
L(\PERP)}
\\ & = \; \frac{L^* \cdot \SetCard{\Pi^*}}{L^* \cdot \SetCard{\Pi^*}
+ \sum_{\PERP \notin \Pi^*} L(\PERP)}
\\ & \geq \; \frac{L^* \cdot \SetCard{\Pi^*}}{L^* \cdot \SetCard{\Pi^*}
     + n! \cdot L^* \cdot \frac{1}{(n + 1)!}}
\\ & = \; \frac{\SetCard{\Pi^*}}{\SetCard{\Pi^*} + 1/(n + 1)}
\\ & \geq \; 1 - 1/(n+1).
\end{align*}

If \NodeWeightP has total variation distance at most
$1 - \TVDis$ from \NodeWeight, it must satisfy
$\sum_{\PER \in \Pi^*} \NodeWeightP[\PER]
\geq \sum_{\PER \in \Pi^*} \NodeWeight[\PER] - (1-\TVDis)
\geq \TVDis - 1/(n+1)$.
In particular, it must sample a permutation $\PER \in \Pi^*$ with
constant probability $\TVDis - 1/(n+1)$.
  
A randomized algorithm can now simply sample $O(\log n)$
permutations \PER according to \NodeWeightP.
If one of these permutations, applied to the nodes of $G$,
has at most $k$ edges going from higher-numbered to lower-numbered
nodes, it constitutes a feedback arc set of at most $k$ edges,
and the algorithm can correctly answer ``Yes'' to the
\textsc{Minimum Feedback Arc Set} instance.
When the algorithm sees no \PER with fewer than $k+1$ edges going
from higher-numbered to lower-numbered nodes, it answers ``No.''
This answer may be incorrect.
But notice that if it is incorrect,
the \textsc{Minimum Feedback Arc Set} instance must have had
a feedback arc set of at most $k$ edges,
and the randomized algorithm would sample at least one corresponding
permutation \PER with high probability.
Thus, when the algorithm answers ``No,'' it is correct with high
probability. 
Thus, we have an \Complexity{RP} algorithm for
\textsc{Minimum Feedback Arc Set} under the assumption of an
efficient approximate sampling oracle.
\end{proof}

It should be noted that the value of $p$ in the reduction
is exponentially close to $1$.
In this range, incorrect feedback is so unlikely
that with high probability,
the algorithm will always see a partial order.
It might then still be able to sample efficiently.
On the other hand, for smaller values of $p$ (e.g., constant $p$), 
sampling approximately from the likelihood distribution might be possible
via a metropolized Karzanov-Khachiyan chain or a different approach.
This problem is still open.

\subsection{\AnySwapText Model}

In order to demonstrate that the condition of
Definition~\ref{def:shortest-path} is not trivial to satisfy,
we consider another natural feedback model for ranking.
In the \AnySwap model, a user can exhibit two \emph{arbitrary}
elements $i, j$ that are in the wrong order;
doing so does not imply anything about the relation
between $i, j$ and the elements that are between them.
We will show that in contrast to the \Isort and \Bsort models,
there is \emph{no} almost-undirected graph $G$ satisfying
Definition~\ref{def:shortest-path};
hence, our general framework cannot lead to
an $o(n^2)$ interactive algorithm for learning a ranking
in the \AnySwap model.

\begin{theorem} \label{thm:arbitrary-swap}
For the \AnySwap model, there is no directed graph $G$
which is almost undirected with $c < n$
and satisfies Definition~\ref{def:shortest-path}.
\end{theorem}

\begin{proof}
Assume that there is a graph $G$ which satisfies the definition.
For every $1 \leq i \leq n$, define the permutation
$\PER_i = \Angle{i, i + 1, \ldots, n, 1, \ldots, i - 1}$
and let $\ConsistentSet = \SpecSet{\PER_i}{1 \leq i \leq n}$.
Let \NodeWeight be the node \weightV function
that assigns uniform \weightV to every $\PER_i$
and $0$ to all other permutations.

Consider an arbitrary permutation \PER proposed to the user.
We show that there exists a response that it is consistent with
every permutation in \ConsistentSet but one.
Distinguish the following two cases for the proposed \PER:

\begin{itemize}
\item If there exists $1 \leq i < n$
such that $\PER(i + 1) < \PER(i)$
(i.e., $i + 1$ precedes $i$ in \PER),
then the response ``\emph{$i$ and $i + 1$ are in the wrong order}''
is consistent with every permutation in \ConsistentSet except $\PER_{i + 1}$.
\item If $\PER(i) < \PER(i + 1)$ for every $1 \leq i < n$,
then $\PER = \Angle{1, 2, \ldots, n}$.
In this case, ``\emph{$n$ and $1$ are in the wrong order}''
is a response that
is consistent with every permutation in \ConsistentSet except $\PER_1$.
\end{itemize}

Hence, $\Potential{\NodeWeight}{\PER} \geq \frac{n - 1}{n}$
for every permutation \PER.
This implies that $G$ cannot be almost undirected with $c < n$;
otherwise, Proposition~\ref{prop:approximate-median-exists}
would imply the existence of a permutation \PER with 
$\Potential{\NodeWeight}{\PER} < \frac{n - 1}{n}$.
\end{proof}

While Theorem~\ref{thm:arbitrary-swap}
rules out an algorithm based on
the graph framework we propose,
it is worth noting that
there is an algorithm (not based on our framework)
that, in the absence of noise, can learn the correct permutation
under the \AnySwap model using $O(n \log n)$ queries.
It is an interesting question for future work
to generalize our model so that it contains this algorithm
as a natural special case.


\section{Application II: Learning a Clustering}
\label{sec:clustering}
Many traditional approaches for clustering
optimize an (explicit) objective function
or rely on assumptions about the data generation process.
In interactive clustering,
the algorithm repeatedly proposes a clustering,
and obtains feedback that two proposed clusters should be merged,
or a proposed cluster should be split into two.
There are $n$ items, and a \emph{clustering} \Clustering
is a partition of the items into disjoint sets (\emph{clusters})
$C_1, C_2, \ldots$.
It is known that the target clustering has $k$ clusters,
but in order to learn it, the algorithm can query clusterings with
more or fewer clusters as well.
The user feedback has the following semantics,
as proposed by Balcan and Blum \cite{balcan-blum:2008:split-merge} and
Awasthi et al.~\cite{awasthi-zadeh:2010:supervised-clustering,%
awasthi-balcan-voevodski:2017:local-algorithm-journal}.

\begin{enumerate}
\item $\Merge(C_i, C_j)$: Specifies that all items in $C_i$ and
$C_j$ belong to the same cluster.
\item $\Split(C_i)$: Specifies that cluster $C_i$ needs to be split,
but not into which subclusters.
\end{enumerate}

Notice that feedback that two clusters be merged,
or that a cluster be split (when the split is known),
can be considered as adding constraints on the clustering
(see, e.g., \cite{wagstaff:2002:intelligent-clustering});
depending on whether feedback may be incorrect,
these constraints are hard or soft.

We define a weighted and \emph{directed} graph \UCG on all
clusterings \Clustering.
Thus, $N = B_n \leq n^n$ is the \Kth{n} Bell number.
When \ClusteringP is obtained by a \Merge of two clusters in \Clustering,
\UCG contains a directed edge $(\Clustering, \ClusteringP)$ of \length $2$.
If $\Clustering = \Set{C_1, C_2, \ldots}$ is a clustering,
then for each $C_i \in \Clustering$, the graph \UCG contains
a directed edge of \length $1$ from \Clustering to
$\Clustering \setminus \Set{C_i} \cup \SpecSet{\Set{v}}{v \in C_i}$.
That is, \UCG contains an edge from \Clustering to the clustering
obtained from breaking $C_i$ into singleton clusters of all its elements.
While this may not be the ``intended'' split of the user,
we can still associate this edge with the feedback.

\begin{lemma} \label{lem:unspecified-clustering-graph}
\UCG satisfies Definition~\ref{def:shortest-path} with respect to
\Merge and \Split feedback.
\end{lemma}

\begin{proof}
\UCG has a node for every clustering,
and its edges capture every possible user feedback.\footnote{%
As mentioned before, we \emph{translate}
user feedback of the form \Split
into a request for breaking the cluster into singletons.
From the user's perspective, nothing changes.}

Let $\Clustering = \Set{C_1, \ldots, C_k}$ and
$\ClusteringP = \Set{C'_1, \ldots, C'_{k'}}$
be two clusterings with $k$ and $k'$ clusters, respectively.
We call a cluster $C \in \Clustering$ \emph{mixed}
(with respect to \ClusteringP)
if it contains elements from at least
two different clusters in \ClusteringP.
Let $x_{\Clustering, \ClusteringP}$ ($x$ for short)
be the number of mixed clusters $C \in \Clustering$
with respect to \ClusteringP,
and $y_{\Clustering, \ClusteringP}$ ($y$ for short)
the total number of elements in mixed clusters.
Notice that it is possible that
$x_{\Clustering, \ClusteringP} \neq x_{\ClusteringP, \Clustering}$,
and similarly for $y$.
Define the (asymmetric) distance
$d(\Clustering, \ClusteringP)
= 2 y_{\Clustering, \ClusteringP}
- x_{\Clustering, \ClusteringP} + 2 (k - k')$.
We claim that the length of every shortest path
from \Clustering to \ClusteringP
with respect to the edge \lengths \Weight[e]
is $d(\Clustering, \ClusteringP)$.

First, we show that there exists a path of length
$d(\Clustering, \ClusteringP)$ from \Clustering to \ClusteringP.
Start the path by breaking the $x$ mixed clusters in \Clustering,
using \Split edges of \length $1$ each.
At this point, we have a clustering \ClusteringPP
with $y + k - x$ clusters,
each of which is a subset of one of the clusters in \ClusteringP.
Then, using $y + k - x - k'$ cluster merges (each of edge \length 2),
we obtain \ClusteringP.
The total \length of this path is
$x + 2(y + k - x - k') = d(\Clustering, \ClusteringP)$.

Next, we show that there is no path in \UCG
from \Clustering to \ClusteringP shorter than
the claimed bound of $d(\Clustering, \ClusteringP)$.
We do so by induction on the number of edges in the path.
In the base case of $0$ edges,
$\Clustering = \ClusteringP$, so the claimed bound of
$d(\Clustering, \ClusteringP) = 0$ is a lower bound.
Now consider an edge $(\Clustering, \ClusteringBar)$
which is the first edge on a path
from \Clustering to \ClusteringP.
Let $\bar{k}$ be the number of clusters in \ClusteringBar,
and $x = x_{\Clustering, \ClusteringP}$,
$y = y_{\Clustering, \ClusteringP}$,
$\bar{x} = x_{\ClusteringBar, \ClusteringP}$,
$\bar{y} = y_{\ClusteringBar, \ClusteringP}$.
We distinguish two cases based on the type of edge from
\Clustering to \ClusteringBar.

\begin{itemize}
\item If $(\Clustering, \ClusteringBar)$ is a \Merge edge,
then $\Weight[(\Clustering, \ClusteringBar)] = 2$,
and $\bar{k} = k - 1$.
We distinguish two subcases, based on the two clusters
$C_1, C_2 \in \Clustering$ that were merged:
\begin{enumerate}
\item If $C_1$ or $C_2$ was mixed, or $C_1 \cup C_2$ is not mixed,
then $\bar{x} \leq x$ and $\bar{y} \geq y$
(because merging two clusters
cannot remove any elements from mixed clusters).
In particular, $2\bar{y} - \bar{x} \geq 2y - x$.
\item If neither $C_1$ not $C_2$ was mixed,
but the new cluster $C_1 \cup C_2 $ is mixed,
then $\bar{x} = x + 1$ and
$\bar{y} = y + \SetCard{C_1} + \SetCard{C_2} \geq y + 2$.
Therefore, again
$2\bar{y} - \bar{x} \geq 2y + 4 - (x + 1) \geq 2y - x$.
\end{enumerate}
In either case, $2\bar{y} - \bar{x} \geq 2y - x$, so
\begin{align*}
d(\ClusteringBar, \ClusteringP)
& = 2\bar{y} - \bar{x} + 2 (\bar{k} - k')\\
& \geq  2y - x + 2((k-1) - k')\\
& = d(\Clustering, \ClusteringP) - 2\\
& = d(\Clustering, \ClusteringP)
- \Weight[(\Clustering, \ClusteringBar)].
\end{align*}
\item If $(\Clustering, \ClusteringBar)$ is a $\Split(C)$ edge,
then $\Weight[(\Clustering, \ClusteringBar)] = 1$,
and $\bar{k} = k + \SetCard{C} - 1$.
Again, there can be at most one fewer mixed cluster (namely, $C$).
If $C$ was mixed, then $\bar{x} = x - 1$ and $\bar{y} = y - \SetCard{C}$.
Otherwise, $\bar{x} = x$ and $\bar{y} = y$.
In both cases, we have that
$2\bar{y} - \bar{x} \geq 2y - x - 2\SetCard{C} + 1$.
Thus,
\begin{align*}
d(\ClusteringBar, \ClusteringP)
& = 2\bar{y} - \bar{x} + 2 (\bar{k} - k')\\
& \geq 2y - x - 2\SetCard{C} + 1 + 2(k + \SetCard{C} - 1 - k')\\
& = 2y - x + 2(k - k') - 1\\
& = d(\Clustering, \ClusteringP) - 1\\
& = d(\Clustering, \ClusteringP)
- \Weight[(\Clustering, \ClusteringBar)].
\end{align*}
\end{itemize}
In both cases, we can apply induction to \ClusteringP,
and conclude that there is no path of total \length less than 
$d(\Clustering, \ClusteringP)$ from \Clustering to \ClusteringP.

Finally, we verify that every correct feedback \ClusteringBar
to a queried clustering \Clustering
lies on a path of length $d(\Clustering, \ClusteringOpt)$
from \Clustering to the target clustering $\ClusteringOpt$.
\begin{itemize}
\item If $(\Clustering, \ClusteringBar)$
is a correct user response in the form of \Merge, then
$x_{\ClusteringBar, \ClusteringOpt} = x_{\Clustering, \ClusteringOpt}$
and $y_{\ClusteringBar, \ClusteringOpt} = y_{\Clustering, \ClusteringOpt}$.
However, \ClusteringBar has one fewer cluster than \Clustering,
so
\[
d(\ClusteringBar, \ClusteringOpt)
\; = \; d(\Clustering, \ClusteringOpt) - 2
\; = \; d(\Clustering, \ClusteringOpt)
- \Weight[(\Clustering, \ClusteringBar)].
\]
\item If $(\Clustering, \ClusteringP)$ is a correct response
in the form of \Split
and $C$ is the cluster that needs to be split/broken,
then $x_{\ClusteringBar, \ClusteringOpt} =
x_{\Clustering, \ClusteringOpt} - 1$
and $y_{\ClusteringBar, \ClusteringOpt} =
y_{\Clustering, \ClusteringOpt} - \SetCard{C}$.
Moreover, \ClusteringBar has $\SetCard{C} - 1$ more clusters
than \Clustering.
By applying all these equations,
similar to the earlier calculations,
we get that
\[
d(\ClusteringBar, \ClusteringOpt)
\; =  \; d(\Clustering, \ClusteringOpt) - 1
\; = \; d(\Clustering, \ClusteringOpt) - 
\Weight[(\Clustering, \ClusteringBar)].
\]
\end{itemize}
In both cases, by induction on $d(\Clustering, \ClusteringOpt)$,
we show that every correct user feedback from \Clustering
lies on a path of length
$d(\Clustering, \ClusteringOpt)$ to $\ClusteringOpt$.

Finally, we show that each edge is part of ``short'' cycle.
Consider any two clusterings \Clustering, \ClusteringP
with $k, k'$ clusters, respectively.
By using $k - 1$ \Split operations (of \length $1$ each),
we can first break \Clustering into all singletons.
Then, by using $n - k'$ \Merge operations
(the edges having \length $2$ each),
we can obtain \ClusteringP.
The total \length of this path is at most
$k - 1 + 2(n - k') \leq 3n - 3$.
In particular, for any edge $(\Clustering, \ClusteringP)$,
there is a ``returning'' path from \ClusteringP to \Clustering
of total \length at most $3n - 3$,
which together with the edge $(\Clustering, \ClusteringP)$
gives a cycle of total \length at most $3n - 1 \leq 3n$.
Because $(\Clustering, \ClusteringP)$ has \length at least 1,
it makes up a $\frac{1}{3n}$ fraction of the cycle's \length.
(With a little more care, this bound can be easily improved to
$\frac{1}{2n}$.)
\end{proof}

\UCG is directed, and
(as mentioned in the proof of Lemma~\ref{lem:unspecified-clustering-graph})
every edge makes up at least a
$\frac{1}{3n}$ fraction of the total \length
of at least one cycle it participates in.
Hence, Proposition~\ref{prop:approximate-median-exists}
gives an upper bound of $\frac{3n-1}{3n}$ on the value of $\beta$
in each iteration.
A more careful analysis exploiting the specific structure of \UCG
gives us the following:

\begin{lemma} \label{lem:good-bound-for-clustering}
In \UCG, for every non-negative
node \weightV function \NodeWeight,
there exists a clustering \Clustering with
$\Potential{\NodeWeight}{\Clustering} \leq \half$.
\end{lemma}

\begin{emptyproof}
Without loss of generality,
assume that \NodeWeight is normalized, so that
$\sum\limits_{s} \NodeWeight[s] = 1$.
We describe an explicit greedy procedure for finding
a clustering \Clustering, similar to a procedure employed by
Awasthi and Zadeh~\cite{awasthi-zadeh:2010:supervised-clustering}.
Start with a clustering into singleton sets, i.e., 
$\Clustering = \SpecSet{\Set{v}}{v \text{ is an item}}$.
Repeatedly look for two clusters $C, C' \in \Clustering$
such that the total likelihood of all the clusterings 
that group all of $C \cup C'$ in one cluster is
strictly more than \half.
As long as such $C, C'$ exist, merge them into a new cluster,
and continue with the new clustering.
The procedure terminates with some clustering \Clustering for which
no pair of clusters can be further merged.
We will show that $\Potential{\NodeWeight}{\Clustering} \leq \half$,
by considering all clusterings \ClusteringP adjacent to \Clustering:

\begin{itemize}
\item If \ClusteringP is obtained by merging two clusters
$C, C' \in \Clustering$, then by the termination condition,
$\NodeWeight[\Reach{\Clustering}{\ClusteringP}] \leq \half$;
otherwise, $C$ and $C'$ would have been merged.
\item If \ClusteringP is obtained by splitting a cluster
$C \in \Clustering$,
then we first notice that $C$ cannot be a singleton cluster.
Therefore, it was created by merging two other clusters
at some point earlier in the greedy process.
By the merge condition, the total \weightV of the clusterings
\ClusteringPP that have all of $C$ in the same cluster
is strictly more than \half.
Therefore, the total \weightV of all clusterings \ClusteringPP
that prefer \emph{any} partitioning of $C$ is less than \half.
This implies $\NodeWeight[\Reach{\Clustering}{\ClusteringP}] \leq \half$
in the \Split case as well. \QED
\end{itemize}
\end{emptyproof}

In the absence of noise in the feedback,
Lemmas~\ref{lem:unspecified-clustering-graph} and
\ref{lem:good-bound-for-clustering}
and Theorem~\ref{thm:no-errors} imply an algorithm
that finds the true clustering using
$\log N = \log B(n) = \Theta(n \log n)$ queries.
Notice that this is worse than the ``trivial'' algorithm,
which starts with each node as a singleton cluster and always executes
the merge proposed by the user, until it has found the correct
clustering; hence, this bound is itself rather trivial.

Non-trivial bounds can be obtained
when clusters belong to a restricted set,
an approach also followed by
Awasthi and Zadeh~\cite{awasthi-zadeh:2010:supervised-clustering}.
If there are at most $M$ candidate clusters,
then the number of clusterings is $N_0 \leq M^k$.
For example, if there is a set system \Family of VC dimension
at most $d$ such that each cluster is in the range space of \Family,
then $M = O(n^d)$ by the Sauer-Shelah
Lemma~\cite{sauer:1972:density,%
shelah:1972:combinatorial-problem}.
Combining Lemmas~\ref{lem:unspecified-clustering-graph}
and \ref{lem:good-bound-for-clustering}
with Theorems~\ref{thm:no-errors} and \ref{thm:errors},
we obtain the existence of learning algorithms
with the following properties:

\begin{corollary}
\label{cor:unspecified-clustering-algo}
Assume that in response to each equivalence query,
the user responds with \Merge or \Split.
Also, assume that there are at most $M$ different candidate
clusters, and the clustering has (at most) $k$ clusters.
\begin{enumerate}
\item If all query responses are correct, then the target
clustering can be learned by an interactive algorithm using at most
$\log N = O(k \log M)$ equivalence queries.
Specifically when $M = O(n^d)$, this bound is $O(kd \log n)$.
This result recovers the main
result of \cite{awasthi-zadeh:2010:supervised-clustering}.\footnote{%
In fact, the algorithm in \cite{awasthi-zadeh:2010:supervised-clustering}
is implicitly computing and querying a
node with small \PotentialSym in this directed graph.}
\item If query responses are correct with probability $p > \half$,
the target clustering can be learned with probability at least
$1 - \ErrProb$ by an interactive algorithm using at most
$\frac{(1 - \ErrProb) k \log M}{1 - \Entr{p}} + o(k \log M) + O(\log^2(1/\ErrProb))$
equivalence queries in expectation.
Our framework provides the noise tolerance ``for free;''
\cite{awasthi-zadeh:2010:supervised-clustering}
instead obtain results for a different type of noise
in the feedback.
\end{enumerate}
\end{corollary}

\subsection{Interactive Clustering with Given Cluster Splits}
\label{sec:specified-clusters}

We now also consider a model in which the user specifies exactly
\emph{how} to split a cluster when proposing a split.
The operation $\Split(C_i, C', C'')$ specifies that
the cluster $C_i$ should be split into $C'$ and $C''$,
and thereby implies that none of the items in $C'$
should be clustered with any item in $C''$.
(Naturally, $C'$ and $C''$ must be disjoint,
and their union must be $C_i$.)
We require the same assumptions as for the model of ``unspecified
splits,'' and the bounds we obtain are the same.
Hence, the results in this model are weaker than those for the
``unspecified splits'' model.
We are including them because we believe them to be a clean and
natural application of the interactive learning framework.

We define an undirected weighted graph \GCG,
again containing a node for each clustering \Clustering.
There is an (undirected) edge between two clusterings
$\Clustering, \ClusteringP$ if and only if there exist
clusters $C_i \in \Clustering, C'_j, C'_{j'} \in \ClusteringP$
with
$\ClusteringP = \Clustering
\setminus \Set{C_i} \cup \Set{C'_j, C'_{j'}}$,
i.e., \ClusteringP is obtained from \Clustering
by splitting $C_i$ into $C'_j$ and  $C'_{j'}$.
The \length of the edge $(\Clustering, \ClusteringP)$ is
$\Weight_{(\Clustering, \ClusteringP)} =
2 \SetCard{C'_j} \cdot \SetCard{C'_{j}}$.

\begin{lemma} \label{lem:specified-clustering-graph}
\GCG satisfies Definition~\ref{def:shortest-path}
with respect to \Merge and $\Split(C_i, C', C'')$ feedback.
\end{lemma}

\begin{emptyproof}
Corresponding to each clustering \Clustering,
we define an $n \times n$ adjacency matrix \Adjacency[\Clustering]
with $\Adjacency[\Clustering]_{i, j} = 1$ if items $i$ and $j$
are in the same cluster in \Clustering,
and $\Adjacency[\Clustering]_{i, j} = 0$, otherwise.
(By definition,
$\Adjacency[\Clustering]_{i, i} = 1$ for every item $i$.)
For two clusterings \Clustering and \ClusteringP,
define their distance \Diff{\Clustering}{\ClusteringP}
to be the Hamming distance of their adjacency matrices
\Adjacency[\Clustering] and \Adjacency[\ClusteringP],
i.e., the total number of bits in which their adjacency matrices differ.

If there is an edge in \GCG between \Clustering and \ClusteringP,
then \Clustering and \ClusteringP will differ by exactly one cluster
being split into two (or two clusters being merged into one).
Let $C, C'$ be the two merged clusters
(or the clusters resulting from the split).
The merge/split changes exactly $2 \SetCard{C} \cdot \SetCard{C'}$
bits in the adjacency matrix,
and the \length assigned to the edge $(\Clustering, \ClusteringP)$
is exactly
$\Weight_{(\Clustering, \ClusteringP)}
= 2 \SetCard{C} \cdot \SetCard{C'} = \Diff{\Clustering}{\ClusteringP}$.

We now show that for each pair \Clustering, \ClusteringP,
there is a path in \GCG of total edge \length
exactly equal to \Diff{\Clustering}{\ClusteringP}.
We show this by induction on \Diff{\Clustering}{\ClusteringP},
the base case $\Diff{\Clustering}{\ClusteringP} = 0$ being trivial
because $\Clustering = \ClusteringP$.
Suppose that $\Clustering \neq \ClusteringP$.
Then, there exist\footnote{Technically, one might have to switch the
roles of \Clustering and \ClusteringP for this to be true.}
$C \in \Clustering, C'_1, C'_2 \in \ClusteringP$
such that $C \cap C'_1 \neq \emptyset$ and $C \cap C'_2 \neq \emptyset$.
Consider the move $\Split(C, C \cap C'_1, C \setminus C'_1)$.
Call the resulting clustering \ClusteringPP.
Its adjacency matrix has
$\Adjacency[\ClusteringPP]_{i,j} = 0 = \Adjacency[\ClusteringP]_{i,j}$
for all $i \in C \cap C'_1, j \in C \setminus C'_1$
and $i \in C \setminus C'_1, j \in C \cap C'_1$,
while $\Adjacency[\Clustering]_{i,j} = 1$ for all $i, j \in C$.
Thus,
\begin{align*}
\Diff{\Clustering}{\ClusteringP}
& = \; \Diff{\ClusteringPP}{\ClusteringP}
+ 2\SetCard{C \cap C'_1} \cdot \SetCard{C \cap C'_2}
\; = \; \Diff{\ClusteringPP}{\ClusteringP}
+ \Weight_{(\Clustering,\ClusteringPP)}.
\end{align*}
By induction hypothesis, there is a path of total \length
\Diff{\ClusteringPP}{\ClusteringP} from \ClusteringPP to \ClusteringP
in \GCG, which combined with the edge $(\Clustering, \ClusteringPP)$
gives the desired path from \Clustering to \ClusteringP,
completing the inductive proof.

We can now show that when a user correctly proposes a move
corresponding to an edge $(\Clustering, \ClusteringP)$,
it indeed lies on a shortest path from \Clustering to \OptClustering.
We consider two cases:
\begin{itemize}
\item The user proposes $\Merge(C, C')$. This means that
all of $C$ and $C'$ belong to one cluster in \ClusteringOpt;
in particular, all matrix entries for $i, j \in C \cup C'$ are 1.
In \Adjacency[\Clustering], all entries for $i, j \in C$ are 1,
as are all entries for $i, j \in C'$; on the other hand,
all entries $\Adjacency[\Clustering]_{i,j}$
for $i \in C, j \in C'$ or for $i \in C', j \in C$ are $0$.
After the \Merge, all these entries are $1$ as well,
decreasing the Hamming distance
by $2\SetCard{C} \cdot \SetCard{C'}$.
Since this is the \length of the edge $(\Clustering, \ClusteringP)$ as well, 
and there is a path from \ClusteringP to \ClusteringOpt
with total \length equal to their Hamming distance
(the argument of the previous paragraph),
\ClusteringP indeed lies on a shortest path
from \Clustering to \OptClustering.
\item The user proposes $\Split(C, C', C'')$.
This means that no pair $i \in C', j \in C''$
belongs to the same cluster in \ClusteringOpt,
whereas they are all grouped together in \Clustering,
meaning that $\Adjacency[\Clustering]_{i, j} = 1$
for all $i, j \in C$.
Thus, $\Diff{\ClusteringP}{\ClusteringOpt} =
\Diff{\Clustering}{\OptClustering}
- 2 \SetCard{C'} \cdot \SetCard{C''} =
\Diff{\Clustering}{\ClusteringOpt}
- \Weight_{(\Clustering,\ClusteringP)}$.
Again by the argument from the previous paragraph, 
there is a path of total \length \Diff{\ClusteringP}{\ClusteringOpt}
from \ClusteringP to \ClusteringOpt, so $(\Clustering, \ClusteringP)$
indeed lies on a shortest path from \Clustering to \ClusteringOpt.\QED
\end{itemize}
\end{emptyproof}

As in the clustering model in Section~\ref{sec:clustering},
the obvious $\log N = \log B(n) = \Theta(n \log n)$ bound
on the number of queries can be improved
when clusters belong to a restricted set of size at most $M$,
giving us the following result:

\begin{corollary} \label{cor:specified-clustering-algo}
Assume that in response to each equivalence query,
the user responds with \Merge or $\Split(C_i, C', C'')$.
Also, assume that there are at most $M$ different candidate clusters,
and the clustering has $k$ clusters.
\begin{enumerate}
\item If all query responses are correct, then the target
clustering can be learned by an interactive algorithm using at most
$O(k \log M)$ equivalence queries.
Specifically when $M = O(n^d)$, this bound is $O(k d \log n)$.
\item If query responses are correct with probability $p > \half$,
the target clustering can be learned with probability at least
$1 - \ErrProb$ by an interactive algorithm using at most
$\frac{(1 - \ErrProb) k \log M}{1-\Entr{p}} + o(k \log M) + O(\log^2(1/\ErrProb))$
equivalence queries in expectation.
\end{enumerate}
\end{corollary}

We saw earlier that a trivial
algorithm in the weaker model achieved a bound of $n-k$ queries.
The situation is even more extreme with more informative feedback: 
if there are no errors in the feedback, a bound of $k - 1$ can
actually be obtained by another trivial algorithm.
The algorithm starts from a clustering of all items in one cluster,
and repeatedly obtains feedback,
which must provide a correct split of one cluster into two.
While this algorithm uses significantly fewer queries,
it is not clear how to generalize it to the case of incorrect feedback:
in particular, repeating the same query multiple times
to obtain higher assurance will not work,
as there are many correct answers the algorithm could receive.
Thus, the algorithm cannot rely on a majority vote.

\Eat{Corollary~\ref{cor:specified-clustering-algo}
follows more straightforwardly from the framework:
because this graph is undirected,
one gets $\beta = \half$ immediately.}

\section{Application III: Learning a Classifier} \label{sec:classification}

Learning a binary classifier is the original
and prototypical application
of the equivalence query model
of Angluin~\cite{angluin:1988:queries-concept},
which has seen a large amount of follow-up work since
(see, e.g., \cite{maass-turan:1992:lower-bound-methods,%
maass-turan:1994:geometrical-concepts}).
Naturally, if no assumptions are made on the classifier,
then $n$ queries are necessary in the worst case.
In general, applications therefore restrict the concept classes to
smaller sets, such as assuming that they have bounded VC dimension.
We use \Family to denote the set of all possible concepts,
and write $M = \Size{\Family}$;
when \Family has VC dimension $d$,
the Sauer-Shelah Lemma~\cite{sauer:1972:density,%
shelah:1972:combinatorial-problem}
implies that $M = O(n^d)$.

Learning a binary classifier for $n$ points is an almost
trivial application of our framework\footnote{The results extend
readily to learning a classifier
with $k \geq 2$ labels.}.
When the algorithm proposes a candidate classifier,
the feedback it receives is a point with a corrected label
(or the fact that the classifier was correct on all points).

We define the graph \CLG to be the $n$-dimensional
hypercube%
\footnote{When there are $k$ labels, \CLG is a graph with $k^n$ nodes.}
with unweighted and undirected edges between every pair of nodes at
Hamming distance 1.
Because the distance between two classifiers
\Classifier, \ClassifierP is exactly the number of points on
which they disagree,
\CLG satisfies Definition~\ref{def:shortest-path}.
Hence, we can apply Corollary~\ref{cor:no-errors:half}
and Theorem~\ref{thm:errors} with 
\InitSet equal to the set of all $M$ candidate classifiers 
to obtain the following:

\begin{corollary}
\begin{itemize}
\item With perfect feedback, the target classifier is learned using
$\log M$ queries%
\footnote{With $k$ labels, this bound becomes $(k-1) \log M$.}.
\item When each query response is correct with probability $p > \half$,
there is an algorithm learning the true binary classifier with
probability at least $1 - \ErrProb$ using at most
$\frac{(1 - \ErrProb) \log M}{1 - \Entr{p}} +
o(\log M) + O(\log^2(1/\ErrProb))$
queries in expectation.
\end{itemize}
\end{corollary}

Thus, we recover the classic result on learning a classifier
in the equivalence query model
when feedback is perfect and extend it to the noisy setting.

\subsection{Proper Learning of Hyperplanes}

While the learning algorithm we described will always terminate
with the correct classifier, and in particular one from \Family,
as part of the learning process, it may propose classifiers
outside of \Family, somewhat akin to improper learning.
Angluin's original paper \cite{angluin:1988:queries-concept}
already observed that when each query has to be in \Family,
a large number of queries may be necessary even when \Family
has very small VC dimension:
a particularly stark example is when \Family
consists of all singleton sets.

Since bounded VC dimension is not sufficient
to ensure query-efficient proper learning
in the equivalence query model,
a large body of subsequent work
(see, e.g., \cite{maass-turan:1994:geometrical-concepts}
for an overview)
has focused on specific geometric classes
such as hyperplanes or axis-aligned boxes.
Here, we show that the results on learning hyperplanes
in $\R^d$ can be obtained directly in our framework.
Let \Hyperplanes be the family of all subsets of the $n$ points
that are separable using a $d$-dimensional hyperplane.
The key insight is:

\begin{lemma} \label{lem:hyperplane-classifier}
For every \weightV function
$\NodeWeight : \Hyperplanes \to \R_{\geq 0}$
on linear classifiers,
there exists a linear classifier
$\Classifier \in \Hyperplanes$
with $\Potential{\NodeWeight}{\Classifier} \leq \frac{d + 1}{d + 2}$.
\end{lemma}

The key to the proof of Lemma~\ref{lem:hyperplane-classifier}
is the following generalization of Carath\'{e}odory's Theorem.
We suspect that Lemma~\ref{lem:caratheodory-generalization}
must be known,
but since we could not find a statement of it despite a long search, 
we provide a self-contained proof here.

\begin{lemma} \label{lem:caratheodory-generalization}
Let $P, Q$ be sets in $\R^d$ whose convex hulls intersect.
Then, there exist sets $P' \subseteq P, Q' \subseteq Q$
with $\SetCard{P'} + \SetCard{Q'} \leq d + 2$
such that their convex hulls intersect.
\end{lemma}

Notice that Carath\'{e}odory's Theorem is
the special case when $\SetCard{P} = 1$.
The lemma can also be easily generalized to points
in the intersection of $k$ sets,
rather than just the special case $k = 2$.
  
\begin{proof}
The proof is quite similar to one of the standard proofs
for Carath\'{e}odory's Theorem,
and based on properties of basic feasible solutions of an LP.

First, to avoid notational inconvenience,
we may assume w.l.o.g.~that $P$ and $Q$ are finite; 
in fact, that $\SetCard{P}, \SetCard{Q} \leq d + 1$.
This is due to Carath\'{e}dory's Theorem.
Let $x \in \conv(P) \cap \conv(Q)$.
Then, because $x \in \conv(P)$,
there exists $P' \subseteq P$ of cardinality at most $d + 1$
such that $x \in \conv(P')$; similarly for $Q$. 
Hence, we have exhibited subsets
$P' \subseteq P, Q' \subseteq Q$
of size at most $d + 1$ each, whose convex hulls intersect.
For the remainder of the proof,
we can therefore focus on $P', Q'$,
and will rename them to $P, Q$.

Write $P = \Set{p_1, p_2, \ldots, p_{d+1}}$ and
$Q = \Set{q_1, q_2, \ldots, q_{d + 1}}$.
A point $x$ is in $\conv(P)$ iff there exist
$\lambda_1, \ldots, \lambda_{d + 1} \geq 0$
with $\sum_i \lambda_i = 1$
such that $x = \sum_i \lambda_i p_i$.
Similarly, $x \in \conv(Q)$ iff there exist
$\mu_1, \ldots, \mu_{d + 1} \geq 0$
with $\sum_j \mu_j = 1$  
such that $x = \sum_j \mu_j q_j$.
We can therefore characterize the intersection
$\conv(P) \cap \conv(Q)$ using the following linear program
with variables $\lambda_i, \mu_j$:

\begin{align*}
\sum_i \lambda_i p_i & = \sum_j \mu_j q_j\\
\sum_i \lambda_i & = 1\\
\sum_j \mu_j & = 1\\
\lambda_i & \geq 0 \quad \mbox{ for all } i\\
\mu_j & \geq 0 \quad \mbox{ for all } j.
\end{align*}

Notice that the first ``constraint''
is actually $d$ constraints, one for each dimension.
Hence, the linear program has $2d + 2$ variables
and $3d + 4$ constraints.
By assumption, this LP has a feasible solution,
so it must have a \emph{basic} feasible solution.
A basic feasible solution is characterized by
$2d + 2$ constraints that
hold with equality.
Therefore, at most $d + 2$ inequalities can be strict.
The only inequalities in the LP are
the non-negativity constraints,
implying that at most $d + 2$ variables $\lambda_i, \mu_j$
can be strictly positive.
Hence, there is a point in the intersection
that can be written as
a convex combination of points $p_i$ and $q_j$,
using at most $d + 2$ points total.
\end{proof} 

Using Lemma~\ref{lem:caratheodory-generalization},
the proof of Lemma~\ref{lem:hyperplane-classifier}
is fairly straightforward.

\begin{extraproof}{Lemma~\ref{lem:hyperplane-classifier}}
We will use the terms ``linear classifier'' and ``hyperplane''
interchangeably in the proof,
using whichever term better emphasizes the concept
we are illustrating at the time. 
W.l.o.g., we assume that the \weightsV \NodeWeight
assigned to linear classifiers are normalized
so that they add up to $1$.

For each linear classifier $\Classifier \in \Hyperplanes$
and sample point $x$, let $\Classifier(x) \in \Set{0, 1}$
be the assigned binary label.
Define $\phi(x) := \sum_{\Classifier \in \Hyperplanes}
\NodeWeight[\Classifier]  \cdot\Classifier(x)$
as the weighted average label assigned to $x$
by all classifiers.
Define $P := \SpecSet{x}{\phi(x) < \frac{1}{d + 2}}$
and $Q := \SpecSet{x}{\phi(x) > 1-\frac{1}{d + 2}}$.

We claim that the convex hulls of $P$ and $Q$ do not intersect.
Suppose for contradiction that they did;
then, by Lemma~\ref{lem:caratheodory-generalization},
there exist $P' \subseteq P, Q' \subseteq Q$
with $\SetCard{P'} + \SetCard{Q'} \leq d + 2$
and $x \in \conv(P') \cap \conv(Q')$.
Write $P' = \Set{p_1, \ldots, p_k}$ and
$Q' = \Set{q_1, \ldots, q_\ell}$
with $k + \ell \leq d + 2$.
For each $p_i$, let
$\HyperplanesS{p}{i} :=
\SpecSet{\Classifier \in \Hyperplanes}{\Classifier(p_i) = 0}$
be the set of all linear classifiers
that assign label $0$ to $p_i$;
similarly, let 
$\HyperplanesS{q}{j} :=
\SpecSet{\Classifier \in \Hyperplanes}{\Classifier(q_j) = 1}$
be the set of all linear classifiers
that assign label $1$ to $q_j$.
By the definition of $P$ and $Q$, 
$\NodeWeight[\HyperplanesS{p}{i}] > 1 - \frac{1}{d + 2}$
for all $i$, and
$\NodeWeight[\HyperplanesS{q}{j}] > 1 - \frac{1}{d + 2}$
for all $j$, or, taking complements,
$\NodeWeight[\Compl{\HyperplanesS{p}{i}}] < \frac{1}{d + 2}$
for all $i$, and
$\NodeWeight[\Compl{\HyperplanesS{q}{j}}] < \frac{1}{d + 2}$
for all $j$.
Because
\begin{align*}
\NodeWeight[\bigcup_i \Compl{\HyperplanesS{p}{i}} \cup
\bigcup_j \Compl{\HyperplanesS{q}{j}}]
& \leq \sum_i \NodeWeight[\Compl{\HyperplanesS{p}{i}}]
+ \sum_j \NodeWeight[\Compl{\HyperplanesS{q}{j}}]
\; < \; k \cdot \frac{1}{d+2} + \ell \cdot \frac{1}{d+2}
\; < \; 1,
\end{align*}
there must exist at least one linear classifier
$\Classifier \in \bigcap_i \HyperplanesS{p}{i}
\cap \bigcap_j \HyperplanesS{q}{j}$.
Because $\Classifier(p_i) = 0$ for all $p_i \in P'$,
and $x \in \conv(P')$, we must have $\Classifier(x) = 0$.
But because $\Classifier(q_j) = 1$ for all $q_j \in Q'$,
and $x \in \conv(Q')$, we must also have $\Classifier(x) = 1$.
This is a contradiction, and we have proved that
the convex hulls of $P$ and $Q$ are disjoint.

Because $\conv(P) \cap \conv(Q) = \emptyset$,
the Hyperplane Separation Theorem implies that
there is a hyperplane \Classifier separating $P$ and $Q$.
We will show that any such \Classifier
(labeling every point in $P$ with $0$ every point in $Q$ with $1$)
satisfies the claim of the
lemma; thereto, fix one arbitrarily.
Consider any feedback that could be given to the algorithm;
because the graph \CLG is the $n$-dimensional hypercube,
this feedback is in the form of a point $x$
which \Classifier mislabels. We distinguish three cases for $x$:

\begin{itemize}
\item If $x \in P$, then $\Classifier(x) = 0$.
By definition of $P$, the total weight of classifiers
labeling $x$ with $0$ is more than $1 - \frac{1}{d + 2}$,
and all these classifiers are inconsistent with the feedback.
Therefore, the total weight of all classifiers
consistent with the feedback decreased by a factor $d + 2$,
which is much stronger than the claim of the lemma.
\item If $x \in Q$, then $\Classifier(x) = 1$; apart from this,
the proof is identical to the case $x \in P$.
\item If $x \notin P \cup Q$,
then either $\Classifier(x) = 0$ or
$\Classifier(x) = 1$ are possible.
The fractional label $\phi(x) \in [\frac{1}{d + 2}, 1-\frac{1}{d + 2}]$
was inconclusive.
But because $\phi(x) \geq \frac{1}{d + 2}$,
at least a $\frac{1}{d + 2}$ (weighted) fraction of classifiers 
labeled $x$ with $1$;
similarly, because $\phi(x) \leq 1 - \frac{1}{d + 2}$,
at least a $\frac{1}{d+2}$ (weighted) fraction of classifiers 
labeled $x$ with $0$.
Thus, whichever label \Classifier assigned to $x$
and was corrected about,
at least a $\frac{1}{d + 2}$ weighted fraction of classifiers
are inconsistent with the feedback.
\end{itemize}

Thus, in each case, we obtain that the total weight of classifiers
consistent with the feedback is at most a
$\max(\frac{1}{d + 2}, 1-\frac{1}{d + 2})
= \frac{d + 1}{d + 2}$ fraction of the total weight,
and the lemma follows.
\end{extraproof}

Using Lemma~\ref{lem:hyperplane-classifier},
Theorem~\ref{thm:no-errors} with
$\beta = \frac{d + 1}{d + 2}$,
and the fact that $M \leq n^d$,
we immediately recover Theorem 5
of \cite{maass-turan:1990:complexity}
for binary classification:

\begin{corollary}[Theorem 5 of \cite{maass-turan:1990:complexity}]
In the absence of noise, a hyperplane can be properly learned using 
at most $O(d \log_{\frac{d + 2}{d + 1}} n) = O(d^2 \log n)$
equivalence queries.
\end{corollary}

As with the result for improper learning,
we obtain a bound in the case of imperfect feedback by using
Theorem~\ref{thm:errors} in place of Theorem~\ref{thm:no-errors}.

\section{Discussion and Conclusions} \label{sec:conclusions}

We defined a general framework for interactive learning
from imperfect responses to equivalence queries,
and presented a general algorithm that achieves
a small number of queries.
We then showed how query-efficient interactive learning algorithms
in several domains can be derived with practically no effort
as special cases; these include some previously known results
(classification and clustering) as well as 
new results on ranking/ordering.

Our work raises several natural directions for future work.
Perhaps most importantly, for which domains can the algorithms
be made computationally efficient (in addition to query-efficient)?
We provided a positive answer for ordering with perfect query responses,
but the question is open for ordering when feedback is imperfect.
For classification, when the possible clusters have VC dimension $d$,
the time is $O(n^d)$, which is unfortunately still impractical
for real-world values of $d$.
Maass and Tur\'{a}n \cite{maass-turan:1990:complexity}
show how to obtain better bounds
specifically when the sample points form a $d$-dimensional grid;
to the best of our knowledge, the question is open
when the sample points are arbitrary.
The Monte Carlo approach of Theorem~\ref{thm:main-sampling-theorem}
reduces the question to the question of sampling
a uniformly random hyperplane,
when the uniformity is over the \emph{partition}
induced by the hyperplane (rather than some geometric representation).
For clustering, even less appears to be known.

Another natural question is motivated by the discussion
in Section~\ref{sec:classification}
and its application to clustering.
We saw that the number of queries can increase significantly
in proper interactive learning, i.e.,
when the set of \structures that can be queried is restricted
to those that are themselves candidates.
When concepts were specifically defined as hyperplane partitions,
the increase can be bounded by $O(\log d)$.
Can similar bounds be obtained for other concept classes?
What happens in the case of clustering if each proposed clustering
must only have clusters from the allowed set?

Finally, we are assuming a uniform noise model.
An alternative would be that the probability of an incorrect response
depends on the type of response.
In particular, false positives could be extremely likely, for instance,
because the user did not try to classify
a particular incorrectly labeled data point,
or did not see an incorrect ordering of items
far down in the ranking.
Similarly, some wrong responses may be more likely than others;
for example, a user proposing a
merge of two clusters (or split of one)
might be ``roughly'' correct, but miss out on a few points
(the setting that
\cite{awasthi-zadeh:2010:supervised-clustering,%
awasthi-balcan-voevodski:2017:local-algorithm-journal}
studied).
We believe that several of these extensions should be fairly
straightforward to incorporate into the framework,
and would mostly lead to additional complexity in notation
and in the definition of various parameters.
But a complete and principled treatment would be an
interesting direction for future work.

\bibliographystyle{abbrv}
\bibliography{names,conferences,references}




\end{document}